\documentclass[journal]{IEEEtran}
\usepackage{amsmath,amsthm, amssymb, amsfonts, amsbsy, mathrsfs, bm , bbm}
\usepackage{cite, enumerate}
\usepackage[table]{xcolor}
\usepackage{booktabs}
\usepackage[normalem]{ulem}
\usepackage[sans]{dsfont}
\usepackage{algorithm}
\usepackage{algpseudocode}
\usepackage{subfigure}
\usepackage{graphicx}
\usepackage{epstopdf}
\usepackage{balance}

 \theoremstyle{plain}

\newtheorem{prop}{Proposition}

\theoremstyle{definition}

\theoremstyle{remark}

\newtheorem{remark}{Remark}
\newtheorem{assumption}{Assumption}
\newtheorem{cor}{Corollary}

\def\R{{\mathbb{R}}}

\def\C{{\mathbb{C}}}

\newcommand{\E}[1]{\mathsf{E}\hspace{-0.8mm}\left[#1 \right]}

\newcommand{\Blkdiag}[1]{\mathsf{Blkdiag}\left\{#1 \right\}}
\newcommand{\Diag}[1]{\mathrm{diag}\left\{#1 \right\}}
\newcommand{\vect}[1]{\mathrm{vec}\left\{#1 \right\}}
\newcommand{\tr}[1]{\mathsf{tr}\hspace{-0.8mm}\left(#1 \right)}

\def\T{\scriptscriptstyle{\mathsf{T}}}
\def\H{\scriptscriptstyle{\mathsf{H}}}

\def\M{\scriptscriptstyle{\mathrm{M}}}

\def\argmin{\mathop{\mathrm{arg\,min}}}

\def\bc{{\mathbf{c}}}

\def\bx{{\mathbf{x}}}

\def\bz{{\mathbf{z}}}

\def\bu{{\mathbf{u}}}
\def\bh{{\mathbf{h}}}
\def\bg{{\mathbf{g}}}

\def\bq{{\mathbf{q}}}

\def\b1{{\mathds{1}}}

\def\bw{{\mathbf{w}}}
\def\bsig{{\pmb{\sigma}}}
\def\bga{{\boldsymbol{{\gamma}}}}

\def\u{{\upsilon}}
\def\e{{\varepsilon}}
\def\s{{\sigma}}

\def\bH{{\mathbf{H}}}

\def\bG{{\mathbf{G}}}

\def\bQ{{\mathbf{Q}}}

\def\bI{{\mathbf{I}}}

\def\bP{{\mathbf{P}}}
\def\bK{{\mathbf{K}}}
\def\bS{{\mathbf{S}}}
\def\bT{{\mathbf{T}}}

\def\bC{{\mathbf{C}}}

\def\bD{{\mathbf{D}}}

\def\bmS{{\boldsymbol{\mathcal{S}}}}

\def\bmF{{\boldsymbol{\mathcal{F}}}}

\def\bmJ{{\boldsymbol{\mathcal{J}}}}
\def\bmI{{\boldsymbol{\mathcal{I}}}}
\def\bmT{{\boldsymbol{\mathcal{T}}}}
\def\bmU{{\boldsymbol{\mathcal{U}}}}

\def\bmSig{{\boldsymbol{{\Sigma}}}}
\def\bmGa{{\boldsymbol{{\Gamma}}}}

\begin{document}

\title{Variants of Partial Update Augmented CLMS Algorithm and Their Performance Analysis}

\author{Vahid Vahidpour, Amir Rastegarnia, Azam Khalili, Wael M. Bazzi, and Saeid Sanei, 

\thanks{Manuscript received 2018.
V. Vahidpour, A. Rastegarnia, and A. Khalili are with the Department of Electrical Engineering, Malayer University, Malayer 65719-95863, Iran (email: v.vahidpour@stu.malayeru.ac.ir; rastegarnia@malayeru.ac.ir; a-khalili@tabrizu.ac.ir).

W. M. Bazzi is with the Electrical and Computer Engineering Department, American University in Dubai
Dubai, United Arab Emirates, (email: wbazzi@aud.edu).

S. Sanei is with the School of Science and Technology, Nottingham Trent
University, Clifton Lane, Nottingham, U.K. (e-mail: saeid.sanei@ntu.ac.uk)
}
 \thanks{Digital Object Identifier 2018/XX}
}

\markboth{IEEE Transactions on Signal Processing}%
{Malekian \MakeLowercase{\textit{et al.}}: Bare Demo of IEEEtran.cls for Journals}

\IEEEpubid{0000--0000/00\$00.00~\copyright~2019 IEEE}


\maketitle

\begin{abstract}
Naturally complex-valued information or those presented in complex domain are effectively processed by an augmented
complex least-mean-square (ACLMS) algorithm. In some applications, the ACLMS algorithm may be too computationally- and  memory-intensive to implement. 
In this paper, a new algorithm, termed partial-update ACLMS (PU-ACLMS) algorithm is proposed, where only a 
fraction of the coefficient set is selected to update at each iteration. Doing so, two types of partial-update 
schemes are presented referred to as the sequential and stochastic partial-updates, to reduce computational 
load and power consumption in the corresponding adaptive filter. The computational cost for 
full-update PU-ACLMS and its partial-update implementations are discussed. Next, the steady-state mean and 
mean-square performance of PU-ACLMS for non-circular complex signals are analyzed and closed-form expressions of the 
steady-state excess mean-square error (EMSE) and mean-square deviation (MSD) are given. Then, employing the weighted energy-conservation relation, the EMSE and MSD learning curves are derived. The simulation results are 
verified and compared with those of theoretical predictions through numerical examples.  
\end{abstract}


\IEEEpeerreviewmaketitle

\section{Introduction}
\label{sec:1}

\IEEEPARstart{C}{omplex}-valued adaptive filters are exploited in many practical signal processing applications such as  
  channel estimation \cite{Sayed2008,jahanchahi2014complex}, frequency estimation \cite{xia2011widely,xia2017widely,xia2012adaptive}, and self-interference cancellation \cite{li2018augmented,li2019cost}. The standard complex least-mean-square (CLMS), as the generic extension of  LMS algorithm in the complex domain $\C$ \cite{xia2018performance,mandic2010steady}, is one of the widely-used adaptive signal processing algorithms, because of its 
simplicity and ease of implementation \cite{widrow1975complex,wu2019steady}. In some scenarios, often the source signals are 
non-circular or improper. Recent advances have put this assumptions under scrutiny \cite{mandic2009complex,adali2011complex}. Specifically, adaptive 
filtering techniques developed with non-circularity or impropriety in mind have been manifested to possess superior performance in an expanding number of applications \cite{mandic2009complex,adali2011complex,clark2010multiband,zarei2016q}.

The augmented complex statistics have provided the possibility to adequately use the complementary information 
of signal non-circularity \cite{adali2011complex,schreier2003second,picinbono1995widely,khalili2016quantized}. This has served as a basis for the evolution of the class of augmented 
adaptive filtering algorithms. These adaptive algorithms are usually known as widely linear algorithms, 
e.g., widely linear LMS (Wl-LMS) \cite{schober2004data}, ACLMS algorithms \cite{javidi2008augmented}, augmented affine projection algorithm (AAPA) \cite{xia2010augmented}, 
widely linear recursive least-squares (WL-RLS) \cite{douglas2009widely}, regularized normalized augmented complex LMS (RN-ACLMS) \cite{xia2010regularised}, 
and augmented extended Kalman filter (AEKF) algorithms \cite{goh2007augmented}.

Some adaptive filtering applications, like channel equalization, echo cancellation, and multi-user detection, 
require an adaptive filter with a very large coefficient vector. In such applications, the ACLMS 
algorithm may be too computationally- and memory-intensive to implement. In order to overcome the mentioned 
constraints, one might allow a subset of the adaptive filter coefficients to be updated at each iteration, 
rather than the entire coefficient vectors. Such a process is called partial coefficient update or briefly partial update \cite{Dogancay2008}. Availability of a finite number of hardware multipliers often driven by cost, space and power consumption considerations, is the main reason for partial updating. 

\IEEEpubidadjcol

To reduce the computational costs and power consumption, various types of partial update schemes, such as Periodic and Sequential LMS  algorithm \cite{douglas1997adaptive}, and stochastic partial updating \cite{godavarti2005partial} have been applied to LMS algorithm. In the Periodic LMS algorithm, all the filter coefficients are updated every $P$-th iteration. The Sequential LMS algorithm updates only a portion of coefficients
at each iteration. The stochastic partial-update LMS \cite{godavarti2005partial} is a randomized version of sequential LMS algorithm in that the coefficient 
subsets are chosen in a random instead of deterministic fashion. Another approach referred to as max 
partial-update LMS algorithm has been proposed in \cite{aboulnasr1997selective,douglas1996analysis,douglas1994family}. Diniz and Werner \cite{werner2004partial} proposed another variant known 
as set-member-ship partial-update NLMS algorithm based on data-selective updating. The performance analysis 
of time-domain adaptive filters in the under-modeling situation is established in \cite{xia2018performance}, where a deficient length 
ACLMS in $\C$ has been considered for second order improper signals. Some distributed versions of partial-update 
adaptive filters such as \cite{Arablouei2014,Arablouei2014a,vahidpour2018analysis,vahidpour2017analysis,vahidpour2019partial,vahidpour2019performance} have been developed in the literature.

In this paper,  a reduced complexity ACLMS algorithm employing partial updating for improper
complex signals is proposed. The algorithm, referred to as partial-update ACLMS (PU-ACLMS) algorithm involves selection of a  fraction of the coefficients at every iteration. To this end, we consider two types partial-update schemes namely sequential and stochastic partial-update. 

The main contributions of this paper are summarized as follows:
\begin{itemize}
	\item A new algorithm, called PU-ACLMS is proposed to control the computational complexity;
	\item The computational complexity for full-update ACLMS and its partial-update 
		implementations are examined thoroughly. For large filter lengths the PU-ACLMS algorithms are 
		able to lower the full-complexity by approximately a factor of two;
	\item In the absence of exact performance analysis the concept of energy-conservation is 
	utilized to derive approximate closed-form expressions for mean-square-error (MSE) and excess mean-square-error 
	(EMSE) of the proposed algorithm;
	\item The closed-form expressions enable us to find a monotonically increasing relationship between these 
	quantities and the step-size parameter $\mu$;
	\item The stability conditions for PU-ACLMS algorithm are derived both in mean and mean-square 
	senses for non-circular signals scenarios. Conditions on step-size $\mu$ are established to guarantee the mean and mean-square stability 	of PU-ACLMS algorithms;
	\item Employing the energy-conservation approach, closed-form expressions to 
	describe the EMSE and MSD learning curves are derived;
	\item The convergence rates of PU-ACLMS algorithms and the full-update ACLMS are investigated.
\end{itemize}

Throughout the paper, we adopt normal lowercase letters for scalars, bold lowercase letters for column 
vectors and bold uppercase letters for matrices, while $\bI$ denotes an identity matrix of appropriate 
size. The real and complex domains are denoted by $\R$ and $\C$. $\R^{\geq 0}$ denotes the set of positive 
real numbers. For ease of reference, a list of main symbols used throughout the text are collected in 
Table \ref{tbl:1}.
\begin{table}[]
\caption{Symbols and Their Descriptions}
	\centering
		\begin{tabular}{lcl}
		\hline
		Symbol & \hspace{0.2cm} & Description \\
				\hline \hline \\
				\vspace{-5.5mm}
				\\
		$(\cdot)^{\T}$	 & \hspace{0.2cm} & Matrix transposition \\
		$(\cdot)^*$	& \hspace{0.2cm} & Complex conjugate    \\
		$(\cdot)^{\H}$	& \hspace{0.2cm} & Hermitian transposition  \\
		$\|\cdot\|^2 $	& \hspace{0.2cm} & Squared Euclidean norm \\
		$\|\bx\|^2_{\bmSig} $	& \hspace{0.2cm} & Weighted norm, $\bx^{*}\bmSig \bx$    \\
		$\mathrm{tr} $	& \hspace{0.2cm} & Trace of a matrix    \\
		$\E{\cdot} $	& \hspace{0.2cm} & Statistical expectation  \\
		$\otimes $	& \hspace{0.2cm} & Kronecker product  \\
		\hline 
		\end{tabular} 
		\label{tbl:1}
\end{table}

The rest of the paper is organized as follows:
In Section \ref{sec:2}, the widely linear model and the ACLMS algorithm are briefly introduced. The PU-ACLMS algorithm and its computational complexity is provided in Section \ref{sec:3}. Section \ref{sec:4} investigates the different aspects of PU-ACLMS algorithm including its steady-state performance, stability conditions, transient performance and convergence analysis. Performance evaluations are illustrated in Section \ref{sec:5}. The paper is finally concluded in Section \ref{sec:6}.

\section{Augmented CLMS Algorithm}
\label{sec:2}
Consider the parameter estimation problem as depicted in Fig. \ref{fig:1}. Here, $d(n)\in \C$ denotes a second order non-circular desired signal  generated by a widely linear model as
\begin{equation}
\label{eq:1}
d(n) = \bu^{\T}(n)\bh^{o}+\bu^{\H}(n)\bg^{o}+\u(n)
\end{equation}
where $\bh^{o}$ and $\bg^{o}$ denote the optimal standard and conjugate 
weight vectors respectively. Moreover, $\u(n) \in\C$ is the measurement noise and 
$\bu(n)=[u(n),\ldots,u(n-N+1)]^{\T}\in\C^{N\times 1}$ is the input vector. 
\begin{figure} [t]
\centering \includegraphics [width=7.5cm]{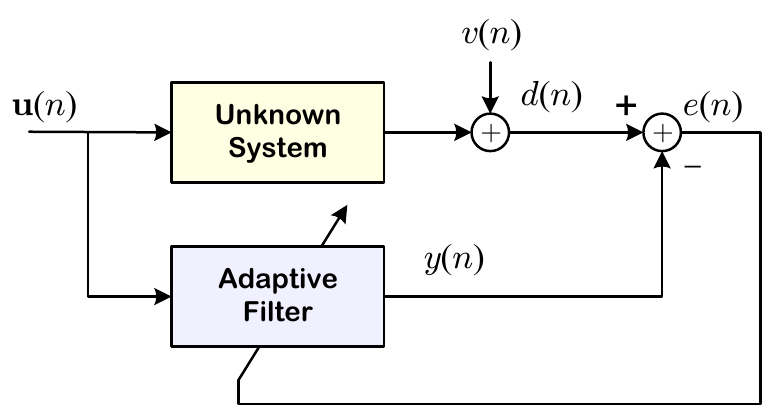} 
\centering \caption{Schematic diagram for a conventional adaptive filter parameter estimation.}
\label{fig:1}
\end{figure}
The following assumptions are considered for the data:
\begin{assumption}\label{asp:1}\
\begin{enumerate}[(i)]
\item The input vectors $\{\bu(n)\}$ and the additive noise $\{\u(n)\}$ are stationary and zero-mean.
\item The noise sequence $\{\u(n)\}$ is independent and identically distributed (i.i.d.) with 
variance $\sigma_{\u}^{2}=\E{|\u(n)|^2}$.
\item The noise sequence $\{\u(n)\}$ is statistically independent of $\bu(\ell)$ for 
all $n\neq \ell$.
\item The regressor covariance matrix is positive-definite $\bC_{\bu} = \E{\bu(n)\bu^{\H}(n)}>0$.
\end{enumerate}
\end{assumption}

The problem of estimating the model parameters can be formulated  as follows:
\begin{equation}
\label{eq:3}
\argmin_{\{\bh,\bg\}}J(\bh,\bg) = \E{|e(n)|^{2}}=\E{|d(n)-y(n)|^2}
\end{equation} 
where $\bh(n)$ and $\bg(n)$ are adjustable filter weight vectors, called standard and 
conjugate weight vectors, respectively. In addition, using the augmented statistics, the output $y(n)$  can be 
written as \cite{mandic2010steady}
\begin{equation}
\label{eq:2}
y(n)= \bu^{\T}(n)\bh(n)+\bu^{\H}(n)\bg(n)
\end{equation}
In order to solve \eqref{eq:3}, the ACLMS algorithm updates its weight vectors according to 
\begin{align}
\label{eq:4}
\bh(n+1)&=\bh(n)+\mu e(n)\bu^{*}(n)\\
\bg(n+1)&=\bg(n)+\mu e(n)\bu(n)
\label{eq:5}
\end{align}
A detailed study of this algorithm can be found in \cite{mandic2010steady}.

\section{Partial-Update ACLMS Algorithm}
\label{sec:3}
\subsection{Algorithm Derivation}
In the partial-update schemes, only $M$ out of $N$  weights are allowed to be updated at each iteration. This can be achieved by modifying the adaptation recursions in \eqref{eq:4} and \eqref{eq:5} as:
\begin{align}
\bh(n+1)&=\bh(n)+\mu e(n)\bmI_{\M}(n)\bu^{*}(n)  \label{eq:6a}\\
\bg(n+1)&=\bg(n)+\mu e(n)\bmI_{\M}(n)\bu(n)
\label{eq:6b}
\end{align}
where $\bmI_{\M}(n)$ is an $N \times N$ diagonal coefficient selection matrix defined as: 
\begin{equation*}
\bmI_{\M}(n)= \Diag{i_{1}(n),i_{2}(n),\cdots,i_{N}(n)}
\end{equation*}
The $\bmI_{\M}(n)$  entries satisfy the following constraints:
\begin{equation*}
	i_{k}(n)\in \{0,1\}, \ \ \sum_{k=1}^N i_{k}(n)=M
\end{equation*}
In the \emph{sequential} partial update $i_{k}(n)$ is given by
\begin{align}
\label{eq:8}
i_{k}(n)=\begin{cases}
1 & \text{if}\, k \in \bmS_{\text{mod}\left(n,\beta\right)+1} \\
0 & \text{otherwise}
\end{cases}
\end{align}
where $\beta = \left\lceil N/M\right\rceil$ and the operator, $\text{mod}(n,\beta)$, returns the reminder of the Euclidean division of $n$ by $\beta$. Briefly, at a given iteration $n$, on of the coefficient subsets $\bmS_{t}, \ t=\left\{1,\ldots, \beta\right\}$ is selected deterministically in a round-robin fashion \cite{Dogancay2008}, and the update is performed. 
\begin{remark}
The coefficient subsets $\bmS_{t}$ are not uniquely specified if they meet 
the following conditions \cite{Dogancay2008}:
\begin{enumerate}
	\item $\bigcup_{t=1}^{\beta}\bmS_{t}=\bS,\ \text{where} \ \bS=\left\{1,2,\ldots,N\right\}$;
	\item $\bmS_{t} \cap \bmS_{\ell}=\phi, \ \forall t,\ell \in \left\{1,\ldots,\beta\right\}\ \text{and}\ t\neq \ell$.
\end{enumerate}
\end{remark}

In the stochastic PU-ACLMS, at a given iteration $n$ one of the sets $\bmS_{t}, t=\{1,\ldots, \beta\}$, is selected in random form $\left\{\bmS_{1},\ldots,\bmS_{\beta}\right\}$ with equal probability, whereas for sequential PU-ACLMS one of the sets 
$\bmS_{t}$ is selected in a deterministic fashion.

\begin{remark}
Let's define the partial-augmented weight vector $\bw(n)$, the augmented system input vectors $\bz(n)$ and 
diagonal matrix $\bmJ_{\M}(n)$ as follows
\begin{align} \label{defs}
\bw(n)& \triangleq\left[\bh^{\T}(n),\bg^{\T}(n)\right]^{\T}\\
\bz(n)&\triangleq \left[\bu^{\H}(n),\bu^{\T}(n)\right]^{\T}\\
\bmJ_{\M}(n)&\triangleq \Blkdiag{\bmI_{\M}(n),\bmI_{\M}(n)}
\end{align}
Using the above definitions, the augmented partial-coefficient-update equations  \eqref{eq:6a} and \eqref{eq:6b} become:
\begin{align}
\label{eq:weq}
\bw(n+1)=\bw(n)+\mu e(n)\bmJ_{\M}(n)\bz(n)
\end{align}
\end{remark}

\begin{remark}
It should be noted that \eqref{eq:6a} and \eqref{eq:6b} can be rewritten in an equivalent form as
\begin{align}
\bh_{\M}(n+1)=\bh_{\M}(n)+\mu e(n)\bu_{\M}^{*}(n)  \label{eq:7a}\\
\bg_{\M}(n+1)=\bg_{\M}(n)+\mu e(n)\bu_{\M}(n)
\label{eq:7b}
\end{align}
where $\bh_{\M}(n)$ denotes $M \times 1$ sub-vector of $\bh(n)$  which is formed at time $n$ by stacking the elements of $\bh(n)$ with $i_{k}(n)=1$. The vectors $\bg_{\M}(n)$ and $\bu_{\M}(n)$ are respectively the $M \times 1$ sub-vector of $\bg(n)$ and $\bu(n)$ defined similarly to $\bh_{\M}(n)$.  It is noteworthy to mention that \eqref{eq:7a} and \eqref{eq:7b} represent a reduced-size adaptive filter in accordance with the partial update.
\end{remark}
The summary of PU-ACLMS algorithm can be seen in Algorithm 1. 

\begin{algorithm}[!t]\footnotesize{
  \caption{The pseudocodes of PU-ACLMS Algorithm.}
  \begin{algorithmic} 
	\State Set a small $\mu$ value and initialize $\bh_{0}$ and $\bg_{0}$ randomly. 
	Then perform the following steps for $n\geq 1$:
		       \State 1. Evaluate the output $y(n)$ via \eqref{eq:2};
         \State 2. Compute the error signal $e(n)=d(n)-y(n)$;
					\State 3. Select $\bmS_{t}$ (based on the sequential or stochastic rules).
			    \State 4. Update $\bh(n)$ via \eqref{eq:6a};
	        \State 5. Update $\bg(n)$ via \eqref{eq:6b};
        \State 6. If convergence is achieved stop, otherwise go to 1.
  \end{algorithmic}}
  \label{alg-1}
  \end{algorithm}

\subsection{Computational Complexity}
\label{sub:2C}
The total computational complexity of the ACLMS algorithm is $16N+2$ real multiplications and $16N$ real additions per iteration. The computational complexity for the \textit{sequential} PU-ACLMS comprises of the following steps at each iteration:
\begin{enumerate}
	\item $8N$ real multiplications and $8N-2$ real additions to compute the current output $y(n)$.
	\item Two real additions are required to evaluate the error signal, $e(n)$.
	\item Two real multiplications are needed to compute $\mu e(n)$.
	\item $8M$ real multiplications and $4M$ real additions are needed to compute $\mu e\left(i\right)\bu_{\M}^{*}\left(i\right)$ and $\mu e(n)\bu_{\M}e(n)$.
	\item Finally, the update of $\bh_{\M}$ and $\bg_{\M}$  requires $4M$ real additions in total.
	\end{enumerate}
\begin{remark}
Clearly, partial update algorithms are likely to reduce the computational cost. Their implementation require some additional computations for selecting the sub-set of weights to  be updated. To address this issue, some available methods such as short-sort approach \cite{Naylor03} that use efficient weight selection procedures to minimize the processing overhead  can be applied. However, they are effective where the unknown system's impulse response is sparse \cite{Naylor05} or the weights are significantly different in value (e.g. amplitude of some weights are much larger than the others). 
\end{remark}

Therefore, each iteration of sequential PU-ACLMS requires $8(N+M)+2$ real multiplications and $8(N+M)$ real additions.

Compared to sequential PU-ACLMS, the computational cost of the stochastic PU-ACLMS algorithm 
requires evaluation of two additional quantities:
\begin{enumerate}
	\item One real multiplication and one real addition for the implementation of random coefficient selection, employing a simple random generator, e.g., linear congruential generator defied by \cite{press2007numerical}:
	\begin{equation}
	\label{eq:9}
	x(n+1) = \text{mod}\{(ax(n)+b),c\}, \ n\geq 0
	\end{equation}
	where $a$ and $b$ are some positive integers, $c$ is the modulus, and $x(0)$ is known as a seed for the random number generator.
	\item One real multiplication and one real addition for the random integer
	\begin{equation}
	\label{eq:10}
	\pi(n) = \frac{\beta-1}{c-1}x(n)+1
	\end{equation}
	where $\pi(n)$ is an independent random variable with discrete uniform distribution:
	\begin{equation}
	\label{eq:11}
	\text{Pr}\{\pi(n)=t\}=\frac{1}{\beta}
	\end{equation}
	\end{enumerate}
Consequently, the stochastic PU-ACLMS algorithm has the computational complexity of $8(N+M)+4$ real 
multiplications and $8(N+M)+2$ real additions per iteration. Table \ref{tbl:2} summarizes the 
respective computational complexity of the ACLMS algorithm and its partial-update implementations, per 
iteration for complex-valued data.
\begin{remark}
\label{rmk:3}
For large values of $N$, the PU-ACLMS algorithm is capable of lowering the full-complexity approximately 
by a factor of two. It can, however, lead to some performance deterioration. 
\end{remark}
\begin{table}[t]
\caption{Computational Complexity of ACLMS and PU-ACLMS Algorithms at Each Iteration.}
	\centering
		\begin{tabular}{lcc}
		\hline
		Algorithm & $\times$ & $+$ \\
				\hline 
		ACLMS	& $16N+2$ & $16N$ \\
		Sequential PU-ACLMS	& $8(N+M)+2$ & $8(N+M)$    \\
		Stochastic PU-ACLMS	& $8(N+M)+4$ & $8(N+M)+2$  \\
				\hline 
		\end{tabular} 
		\label{tbl:2}
\end{table}
\section{Performance Analysis}
\label{sec:4}
In this section, the performance of PU-ACLMS algorithm is studied in detail. More specifically, we investigate the steady-state performance, stability conditions, transient behavior  and the algorithm convergence are investigated.

\subsection{Steady-state Analysis}
In this sub-section, we use PU-ACLMS given by \eqref{eq:7a} and \eqref{eq:7b} to obtain a closed form for EMSE measure   defined as
\begin{align}
\label{eq:29}
\zeta(\infty) \triangleq \lim_{n \to \infty} \E{\left|e_{a}(n)\right|^2}
\end{align}
where $e_{a}(n)$ is defined as
\begin{equation}
e_{a}(n)=\bu^{\T}(n)\widetilde{\bh}(n)+\bu^{\H}(n)\widetilde{\bg}(n)
\end{equation}
with 
\begin{equation}
\widetilde{\bh}(n)=\bh^{o}-\bh(n),\ \ \widetilde{\bg}(n)=\bg^{o}-\bg(n)
\end{equation}
Let's also denote by $\bh_{\M}^{o}(n)$ and $\bg_{\M}^{o}(n)$ $M\times 1$ vectors which are formed at any time instant $n$ by stacking the elements of $\bh^o$ and $\bg^o$  with $i_{k}(n)=1$, respectively. Subtracting \eqref{eq:7a} and \eqref{eq:7b} from $\bh_{\M}^{o}(n)$ and $\bg_{\M}^{o}(n)$ gives
\begin{align}
\label{eq:12}
\widetilde{\bh}_{\M}(n,n+1)=\widetilde{\bh}_{\M}(n,n)-\mu e(n)\bu_{\M}^{*}(n)\\
\widetilde{\bg}_{\M}(n,n+1)=\widetilde{\bg}_{\M}(n,n)-\mu e(n)\bu_{\M}(n)
\label{eq:13}
\end{align}
where $\widetilde{\bh}_{\M}(n,n+1)=\bh_{\M}^{o}(n)-\bh_{\M}(n+1)$ and $\widetilde{\bg}_{\M}(n,n+1)=\bg_{\M}^{o}(n)-\bg_{\M}(n+1)$ are the partial-update coefficient errors at time instant $n$. Pre-multiplying both sides of \eqref{eq:12} by $\bu_{\M}^{\T}(n)$ and both sides of \eqref{eq:13} by $\bu_{\M}^{\H}(n)$ to yield 
\begin{align}
\label{eq:14}
\bu_{\M}^{\T}(n)\widetilde{\bh}_{\M}(n,n+1)&=\bu_{\M}^{\T}(n)\widetilde{\bh}_{\M}(n,n)-\mu \|\bu_{\M}(n)\|^{2} e(n)\\
\bu_{\M}^{\H}(n)\widetilde{\bg}_{\M}(n,n+1)&=\bu_{\M}^{\H}(n)\widetilde{\bg}_{\M}(n,n)-\mu \|\bu_{\M}(n)\|^{2} e(n)
\label{eq:15}
\end{align}
The partial-update \emph{a priori} estimation error,  $\e_{a}(n)$  and the partial-update \emph{a posteriori} estimation error,  $\e_{p}(n)$ are defined as:
\begin{align}
\label{eq:16}
\e_{a}(n)&=\bu_{\M}^{\T}(n)\widetilde{\bh}_{\M}(n,n)+\bu_{\M}^{\H}(n)\widetilde{\bg}_{\M}(n,n)\\
\e_{p}(n)&=\bu_{\M}^{\T}(n)\widetilde{\bh}_{\M}(n,n+1)+\bu_{\M}^{\H}(n)\widetilde{\bg}_{\M}(n,n+1)
\label{eq:17}
\end{align}
It is easy to show that 
\begin{equation}
\label{eq:18}
\e_{p}(n) = \e_{a}(n) - 2\mu e(n)\|\bu_{\M}(n)\|^{2}
\end{equation}
Solving for $e(n)$ we obtain
\begin{equation}
\label{eq:19}
e(n) = \frac{1}{2\mu \|\bu_{\M}(n)\|^{2}}\left[\e_{a}(n)-\e_{p}(n)\right] 
\end{equation}
Substituting \eqref{eq:19} into \eqref{eq:12} and \eqref{eq:13} yields
\begin{align}
&\widetilde{\bh}_{\M}(n,n+1) + \frac{\bu_{\M}^{*}(n)}{2\|\bu_{\M}(n)\|^{2}}\e_a(n)\nonumber \\
&\hspace{2.5cm} =\widetilde{\bh}_{\M}(n,n)+\frac{\bu_{\M}^{*}(n)}{2\|\bu_{\M}(n)\|^{2}}\e_p(n)
\label{eq:20}  \\
&\widetilde{\bg}_{\M}(n,n+1)+ \frac{\bu_{\M}(n)}{2\|\bu_{\M}(n)\|^{2}}\e_a(n)\nonumber \\
&\hspace{2.5cm}  =\widetilde{\bg}_{\M}(n,n)+\frac{\bu_{\M}(n)}{2\|\bu_{\M}(n)\|^{2}}\e_p(n)  \label{eq:21}
\end{align}
Taking the squared Euclidean norm of both sides of \eqref{eq:20} and \eqref{eq:21} the following energy conservation relations are obtained which describe the evolution of the weight error vectors:
\begin{align}
\label{eq:22}
& \left\|\widetilde{\bh}_{\M}(n,n+1)\right\|^{2} + \widetilde{\bh}_{\M}(n,n)\frac{\bu_{\M}^{*}(n)}{2\left\|\bu_{\M}(n)\right\|^{2}}\e_a(n)\nonumber \\
&\hspace{0.5cm} =\left\|\widetilde{\bh}_{\M}(n,n)\right\|^{2}
+\widetilde{\bh}_{\M}(n,n+1)\frac{\bu_{\M}^{*}(n)}{2\left\|\bu_{\M}(n)\right\|^{2}}\e_p(n)
\end{align}
and
\begin{align}
\label{eq:23}
&\left\|\widetilde{\bg}_{\M}(n,n+1)\right\|^{2} + \widetilde{\bg}_{\M}(n,n)\frac{\bu_{\M}^{*}(n)}{2\left\|\bu_{\M}(n)\right\|^{2}}\e_a(n)\nonumber \\
&\hspace{0.5cm} =\left\|\widetilde{\bg}_{\M}(n,n)\right\|^{2}
+\widetilde{\bg}_{\M}(n,n+1)\frac{\bu_{\M}^{*}(n)}{2\left\|\bu_{\M}(n)\right\|^{2}}\e_p(n)
\end{align}
Adding up \eqref{eq:22} and \eqref{eq:23} gives the weight error energy conservation equations 
\begin{align}
\label{eq:24}
& \left\|\widetilde{\bh}_{\M}(n,n+1)\right\|^{2}+\left\|\widetilde{\bg}_{\M}(n,n+1)\right\|^{2}+\frac{\left|\e_{a}(n)\right|^2}{2\left\|\bu_{\M}(n)\right\|^{2}}\nonumber \\
& \hspace{0.5cm} =\left\|\widetilde{\bh}_{\M}(n,n)\right\|^{2}+\left\|\widetilde{\bg}_{\M}(n,n)\right\|^{2}+\frac{\left|\e_{p}(n)\right|^2}{2\left\|\bu_{\M}(n)\right\|^{2}}
\end{align}
The expression \eqref{eq:24} delivers an exact relation between the partial-update a priori and a posteriori estimation errors and partial-update weight errors. At the steady-state we have \cite{Sayed2008}
\begin{align}
\label{eq:25}
\E{\|\widetilde{\bh}_{\M}(n,n+1)\|^{2}}&=\E{\|\widetilde{\bh}_{\M}(n,n)\|^{2}} <\infty\\
\E{\|\widetilde{\bg}_{\M}(n,n+1)\|^{2}}&=\E{\|\widetilde{\bg}_{\M}(n,n)\|^{2}} <\infty
\label{eq:26}
\end{align}
 Taking the statistical expectation of both sides of \eqref{eq:24},  at steady-state \eqref{eq:24} becomes
\begin{equation}
\label{eq:27}
\E{\frac{\left|\e_{a}(n)\right|^2}{2\left\|\bu_{\M}(n)\right\|^{2}}} = \E{\frac{\left|\e_{p}(n)\right|^2}{2\left\|\bu_{\M}(n)\right\|^{2}}} \ \text{as} \ n\rightarrow \infty
\end{equation}
Substituting \eqref{eq:18} into the right-hand-side of above expression, we obtain

\begin{align}
\label{eq:28}
\mu\E{\left|e(n)\right|^{2}\left\|\bu_{\M}(n)\right\|^{2}}=\E{\left|\e_{a}(n)\right|\left|e(n)\right|}\ \text{as} \ n\rightarrow \infty
\end{align}
In the sequel, the steady-state variance relation \eqref{eq:28} is used to obtain a closed-form expression for the EMSE measure. First, the following assumption is considered which is commonly used to study adaptive filters \cite{Sayed2008}:
\begin{assumption}\label{asp:2}\
\begin{enumerate}[(i)]
\item The a priori estimation error $\e_{a}(n)$ is statistically independent of the input regressor vector $\bu_{\M}(m)$ for all $n\neq m$.
\item (separation assumption \cite{Sayed2008}) $\bu_{\M}(n)$ is statistically independent of $e_{a}(n)$ and alternatively. This means that $\left\|\bu_{\M}(n)\right\|^{2}$ is also statistically independent of $e(n)$.
\end{enumerate}
\end{assumption}
Since $e(n)=e_{a}(n)+\u(n)$, \eqref{eq:28} can be rewritten as
\begin{align}
\label{eq:32}
&\mu \E{\left|e_{a}(n)\right|^{2}+2\left|e_{a}(n)\right|\left|\u(n)\right|+\left|\u(n)\right|^{2}}   \nonumber\\
& \hspace{1.5cm} =\E{\left|\e_{a}(n)\right|\left|e_{a}(n)\right|+\left|\e_{a}(n)\right|\left|\u(n)\right|}
\end{align}
Under \emph{Assumptions} \ref{asp:1} and \ref{asp:2}.i, the above expression simplifies to 
\begin{align}
\label{eq:33}
\mu \E{\left|e_{a}(n)\right|^{2}\left\|\bu_{\M}(n)\right\|^2}+\mu\s_{\u}^{2}\tr{\bC_{\bu_{\M}}}&=\E{\e_{a}(n)e_{a}(n)}
\end{align}
where 
\begin{align}
\label{eq:34}
\tr{\bC_{\bu_{\M}}}&=\tr{\E{\bu_{\M}(n)\bu_{\M}^{\H}(n)}}=\E{\left\|\bu_{\M}(n)\right\|^{2}}
\end{align}
For the first term in the RHS of \eqref{eq:32} we have
\begin{align}
\label{eq:35}
\E{\e_{a}(n)e_{a}(n)}&=\rho_{\M}\E{\left|e_{a}(n)\right|^{2}}=\rho_{\M}\zeta(\infty)
\end{align}
where $0 < \rho_{\M}\leq 1$ is a constant that shows how much $e_{a}(n)$ is lessened as a consequence of 
partial coefficient updates. Substituting \eqref{eq:34} into \eqref{eq:33} gives:
\begin{align}
\label{eq:36}
\zeta(\infty) = \frac{\mu}{\rho_{\M}}\left(\E{\left|e_{a}(n)\right|^{2}\left\|\bu_{\M}(n)\right\|^2}+\mu\s_{\u}^{2}\tr{\bC_{\bu_{\M}}}\right)
\end{align}
Employing \emph{Assumption} \ref{asp:2}.ii, $\E{\left|e_{a}(n)\right|^{2}\left\|\bu_{\M}(n)\right\|^2}$ can be separated into the product of two expectations as follows:
\begin{align}
\label{eq:39}
\E{\left|e_{a}(n)\right|^{2}\left\|\bu_{\M}(n)\right\|^2}= \E{\left|e_{a}(n)\right|^{2}}\E{\left\|\bu_{\M}(n)\right\|^2} 
\end{align}
Now substituting \eqref{eq:39} into \eqref{eq:36} gives the following steady-state EMSE expression:
\begin{equation}
\label{eq:40}
\zeta(\infty)=\frac{\mu\s_{\u}^{2}\tr{\bC_{\bu_{\M}}}}{\rho_{\M}-\mu \tr{\bC_{\bu_{\M}}}}
\end{equation}
\begin{remark}
For sufficiently small step-size $\mu$, at steady-state we expect to have:
\begin{equation}
\label{eq:37a}
\E{\left|e_{a}(n)\right|^{2}\left\|\bu_{\M}(n)\right\|^2} \ll \s_{\u}^{2}\tr{\bC_{\bu_{\M}}}
\end{equation}
and we neglect the term $\E{\left|e_{a}(n)\right|^{2}\left\|\bu_{\M}(n)\right\|^2}$ to derive an closed-form for EMSE as
\begin{align}
\label{eq:38a}
\zeta(\infty) \approx \frac{\mu\s_{\u}^{2}\tr{\bC_{\bu_{\M}}}}{\rho_{\M}}
\end{align}
\end{remark}
\begin{cor}
\label{cor:1}
If all the filter coefficients are updated (i.e. $M=N$, $\bmI_{\M}(n)=\bI_{N}$, $\rho_{\M}=1$ and 
$\bC_{\bu_M}=\bC_{\bu}$), then \eqref{eq:40} changes to 
\begin{align}
\label{eq:41}
\zeta_{\mathrm{ACLMS}}(\infty)=\frac{\mu\s_{\u}^{2}\tr{\bC_{\bu}}}{1-\mu\tr{\bC_{\bu}}}
\end{align}
For small step sizes this can be approximated as
\begin{align}
\label{eq:41b}
\zeta_{\mathrm{ACLMS}}(\infty) \approx \mu\s_{\u}^{2}\tr{\bC_{\bu}}
\end{align}
This expression is the steady-state EMSE for a \emph{full-update} PU-ACLMS algorithm.
\end{cor}
\begin{cor}
\label{cor:2}
For the PU-ACLMS algorithm we have\footnote{This relation holds for stationary input signals.} 
\begin{equation}
\label{eq:cr2_1}
\bC_{\bu_{\M}}=\frac{M}{N}\bC_{\bu},\ \ \rho_{\M}=\frac{M}{N}
\end{equation}
Replacing \eqref{eq:cr2_1} in \eqref{eq:41} we obtain 
\begin{align}
\label{eq:par}
\zeta(\infty)&=\frac{\mu\s_{\u}^{2}\frac{M}{N}\tr{\bC_{\bu_{\M}}}}{\frac{M}{N}-\mu\tr{\bC_{\bu_{\M}}}}  =  \frac{\mu\s_{\u}^{2}\tr{\bC_{\bu}}}{1-\mu\tr{\bC_{\bu}}}
\end{align}
which is the same steady-state EMSE as for the \emph{full-update} PU-ACLMS algorithm. 
\end{cor}
According to corollaries \ref{cor:1} and \ref{cor:2}, the steady-state performance of the PU-ACLMS algorithm can be summarized as the following Proposition.
\begin{prop}\label{prop:1}
Let Assumptions \ref{asp:1} and \ref{asp:2} hold. Then, regardless of the values of $M$, both sequential and stochastic partial-update schemes have the same steady-state EMSE as that of the full-update ACLMS algorithm. 
\end{prop}

\subsection{Stability Analysis}
To gain further insight into the performance of partial-update PU-ACLMS algorithms, we shall continue to examine the stability of PU-ACLMS algorithm. To this end, we consider the update equations \eqref{eq:6a} and \eqref{eq:6b}. Let us also define 
the partial-augmented weight vector $\widetilde{\bw}(n)$ as
\begin{align}
\label{eq:50}
\widetilde{\bw}(n)\triangleq\bw^{o}-\bw(n)
\end{align}
where $\bw^{o}=\left[\bh^{o\T},\bg^{o\T}\right]^{\T}$. Substracting both sides of \eqref{eq:weq} from the optimum solution $\bw^{o}$ gives
\begin{align}
\label{eq:51}
\widetilde{\bw}\left(i+1\right)=\widetilde{\bw}(n)+\mu e(n)\bmJ_{\M}(n)\bz(n)
\end{align}
It is useful to mention that the expression \eqref{eq:51} leaves non-updated subsets of $\tilde{\bw}(n)$ and $\tilde{\bw}(n+1)$ identical. From \eqref{eq:51}, the output-error $e(n)$ can be rewritten as 
\begin{align}
\label{eq:52}
e(n)=\bz^{\H}(n)\tilde{\bw}(n)+\u(n)
\end{align}
Then, the weight-error-recursion in \eqref{eq:51} becomes
\begin{multline}
\label{eq:53}
\widetilde{\bw}(n+1)=\left[\bI-\mu\bmJ_{\M}(n)\bz(n)\bz^{\H}(n)\right]\widetilde{\bw}(n)\\
 -\mu \bmJ_{\M}(n)\bz(n)\u(n)
\end{multline}
In the sequel, recursion \eqref{eq:53} is used to obtain the mean and mean-square stability conditions.

\subsubsection{Mean Stability}
\label{sub:4A}
Applying statistical expressions to the both sides of \eqref{eq:53} and employing \emph{Assumption }\ref{asp:1}, yields
\begin{align}
\label{eq:54}
\E{\widetilde{\bw}(n+1)}&=\left(\bI-\mu\E{\bmJ_{\M}(n)\bz(n)\bz^{\H}(n)}\right)\E{\widetilde{\bw}(n)}\nonumber\\
&=\left(\bI-\mu{\bC^*_{\bz_{\M}}}\right)\E{\widetilde{\bw}(n)}
\end{align}
where $\bC_{\bz_{\M}}$ is the partial-augmented covariance matrix of $\bz(n)$, given by
\begin{align}
\label{eq:55}
\bC_{\bz_{\M}}&\triangleq\E{\bmJ_{\M}(n)\bz(n)\bz^{\H}(n)}\nonumber\\
&=\E{
\begin{bmatrix}
\bmI_{\M}(n)  & 0  \\
0  &  \bmI_{\M}(n)     
\end{bmatrix}
\begin{bmatrix}
\bu(n)\bu^{\H}(n)  & \bu(n)\bu^{\T}(n)   \\
\bu^{*}(n)\bu^{\H}(n) &  \bu^{*}(n)\bu^{\T}(n)     
\end{bmatrix}}\nonumber\\ 
&=\begin{bmatrix}
\bC_{\bu_{\M}}  & \bD_{\bu_{\M}}   \\
\bD_{\bu_{\M}}^{*} &  \bC_{\bu_{\M}}^{*}     
\end{bmatrix} 
\end{align}
The matrices $\bC_{\bu_{\M}}$ and $\bD_{\bu_{\M}}$ are respectively denoted as the partial covariance and partial complementary covariance matrices of $\bu_{\M}(n)$.

We establish the following Proposition which guarantees asymptotic unbiasedness of the PU-ACLMS algorithm.
\begin{prop}
\textbf{\emph{(Mean Stability)}}
\label{prop:2}
Let the widely-linear model \eqref{eq:1} and Assumption \ref{asp:1} hold. Then, the PU-ACLMS algorithm is asymptotically unbiased for any initial condition if, and only if, the positive step-size parameter $\mu$ satisfies
\begin{equation}
\label{eq:56}
0<\mu<\frac{2}{\lambda_{\mathrm{max}}\left(\bC_{\bz_{\M}}\right)}
\end{equation}
where $\lambda_{\mathrm{max}}\left(\bC_{\bz_{\M}}\right)$ is the largest eigenvalue of the partial-augmented matrix $\bC_{\bz_{\M}}$. 
\end{prop}  
\begin{proof}
See Appendix \ref{apx:1}.
\end{proof}

\subsubsection{Mean-square Stability}
\label{sub:4B}
By computing the weighted norm of both sides of \eqref{eq:53}, for an arbitrary Hermitioan weighting matrix  $\bmSig>0$, and applying the expectation operator together with employing Assumption \ref{asp:1}, we arrive at the following weighted variance relation:
\begin{multline}
\label{eq:57}
\E{\left\|\widetilde{\bw}(n+1)\right\|_{\bmSig}^{2}}=\E{\left\|\widetilde{\bw}(n)\right\|_{{\bmSig}^{'}}^{2}} \\
 +\mu^{2}\s_{\u}^{2}\E{\bz^{\H}(n)\bmJ_{\M}(n)\bmSig\bmJ_{\M}(n)\bz(n)}
\end{multline}
\begin{multline}
\label{eq:58}
\bmSig' = \bmSig-\mu\bmSig\bmJ_{\M}(n)\bz(n)\bz^{\H}(n)-\mu\bz(n)\bz^{\H}(n)\bmJ_{\M}(n)\bmSig \\
 +\mu^{2}\bz(n)\bz^{\H}(n)\bmJ_{\M}(n)\bmSig\bmJ_{\M}(n)\bz(n)\bz^{\H}(n)
\end{multline}
Under \emph{Assumptions} \ref{asp:1} and \ref{asp:2}, $\widetilde{\bw}(n)$ is independent of both ${\bmSig}^{'}$ and $\bz(n)$. Thus, we can split $\E{\left\|\widetilde{\bw}(n)\right\|_{{\bmSig'}}^{2}}$ into
\begin{align}
\label{eq:59}
\E{\left\|\widetilde{\bw}(n)\right\|_{{\bmSig'}}^{2}}=\E{\left\|\widetilde{\bw}(n)\right\|_{\mathsf{E}[\bmSig']}^{2}}
\end{align}
with the weighting matrix ${\bmSig}^{'}$ replaced by its mean, denoted by $\bmGa$, i.e. $\bmGa \triangleq \E{\bmSig'}$. 
In this way, the recursions \eqref{eq:57} and \eqref{eq:58} are rewritten as follows:
\begin{multline}
\label{eq:61}
\E{\left\|\widetilde{\bw}(n+1)\right\|_{\bmSig}^{2}}=\E{\left\|\widetilde{\bw}(n)\right\|_{{\bmGa}}^{2}}  \\
 +\mu^{2}\s_{\u}^{2}\E{\left\|\bmJ_{\M}(n)\bz(n)\right\|_{\bmSig}^{2}}
\end{multline}
\begin{multline}
\label{eq:62}
\bmGa = \bmSig -\mu \bmSig \bC_{\bz_{\M}}-\mu \bC_{\bz_{\M}}\bmSig \\
 +\mu^{2}\E{\left\|\bmJ_{\M}(n)\bz(n)\right\|_{\bmSig}^{2}\bz(n)\bz^{\H}(n)}
\end{multline}
By defining $\bga \triangleq \vect{\E{\bmGa}}$ and $\bsig \triangleq \vect{\E{\bmSig}}$  
in addition to applying the Kronecker product notation property\footnote{For any matrices $\mathbf{A}$, $\mathbf{B}$ and $\bmSig$  with compatible dimensions we have
\begin{align*}
\vect{\mathbf{A}\bmSig \mathbf{B}}&=\left(\mathbf{B}^{\T}\otimes \mathbf{A}\right)\vect{\bmSig},
\end{align*}} 
we can modify \eqref{eq:62} as 
\begin{multline}
\label{eq:66}
\vect{\bmGa} = \vect{\bmSig}-\mu \vect{\bmSig \bC_{\bz_{\M}}}-\mu\vect{\bC_{\bz_{\M}}\bmSig} \\
+\mu^{2}\E{\vect{\left\|\bmJ_{\M}(n)\bz(n)\right\|_{\bmSig}^{2}\bz(n)\bz^{\H}(n)}}
\end{multline}
After vectorization, the relations \eqref{eq:61} and \eqref{eq:62} become:
\begin{align}
\label{eq:67}
\E{\left\|\widetilde{\bw}(n+1)\right\|_{\bsig}^{2}}&=\E{\left\|\widetilde{\bw}(n)\right\|_{{\bmF\bsig}}^{2}}+\mu^{2}\s_{\u}^{2}\bc_{\M}^{\T}\bsig
\end{align}
where the coefficient matrix $\bmF$ is given by
\begin{align}
\label{eq:68}
&\bmF  = \bI_{\left(2N\right)^{2}}-\mu\left(\bC_{\bz_{\M}}^{\T}\otimes\bI_{2N}+\bI_{2N}\otimes \bC_{\bz_{\M}}\right) \nonumber \\
& \hspace{0.4cm} +\mu^{2}\E{\left(\bz(n)\bz^{\H}(n)\bmJ_{\M}(n)\right)^{\T}\otimes \left(\bz(n)\bz^{\H}(n)\bmJ_{\M}(n)\right)}
\end{align}
and the vector $\bc_{\M}$ is defined as
\begin{align}
\label{eq:69}
\bc_{\M}=\vect{\E{\bmJ_{\M}(n)\bz(n)\bz^{\H}(n)\bmJ_{\M}(n)}}
\end{align}
which is derived from
\begin{align}
\label{eq:70}
&\E{\left\|\bmJ_{\M}(n)\bz(n)\right\|_{\bmSig}^{2}}=\E{\bz^{\H}(n)\bmJ_{\M}(n)\bmSig\bmJ_{\M}(n)\bz(n)}\nonumber\\
&\hspace{0.5cm}=\tr{\E{\bmJ_{\M}(n)\bz(n)\bz^{\H}(n)\bmJ_{\M}(n)\bmSig}}\nonumber\\
&\hspace{0.5cm}=\tr{\E{\bmJ_{\M}(n)\bz(n)\bz^{\H}(n)\bmJ_{\M}(n)}\bmSig}\nonumber\\
&\hspace{0.5cm}=\text{vec}^{\T}\left\{\E{\bmJ_{\M}(n)\bz(n)\bz^{\H}(n)\bmJ_{\M}(n)}\right\}\bsig\nonumber\\
&\hspace{0.5cm}=\bc_{\M}^{\T}\bsig
\end{align}
Note that, the ${\left(2N\right)^{2}\times \left(2N\right)^{2}}$ coefficient matrix $\bmF$ can be expressed as:
\begin{align}
\label{eq:71}
\bmF(\mu) = \bI-\mu\bP+\mu^{2}\bQ
\end{align}
where $\left(2N\right)^{2}\times\left(2N\right)^{2}$ matrices $\left\{\bP,\bQ\right\}$ are given by 
\begin{align}
\label{eq:72}
\bP &= \bC_{\bz_{\M}}^{\T}\otimes\bI_{2N}+\bI_{2N}\otimes \bC_{\bz_{\M}}\\
\label{eq:73}
\bQ &= \E{\left(\bz(n)\bz^{\H}(n)\bmJ_{\M}(n)\right)^{\T}\otimes \left(\bz(n)\bz^{\H}(n)\bmJ_{\M}(n)\right)}
\end{align}

To derive the stability conditions for the step-size parameter $\mu$, we employ the state-space model developed in \cite{Sayed2008}. Here, we avoid to expose the corresponding argument, except to state that the stability of matrix $\bmF$ ensures the mean-square stability of the PU-ACLMS algorithms. In this way, the following proposition holds.
\begin{prop}
\textbf{\emph{(Mean-square Stability)}}
\label{prop:3}
Let the data $\left\{d(n),\bu(n)\right\}$ satisfy the widely-linear model \eqref{eq:1} and the independence assumption \ref{asp:1} holds. Then, the PU-ACLMS algorithm is mean-square stable if, and only if, the step-size $\mu$ is chosen to satisfy 
\begin{align}
\label{eq:74}
0<\mu<\min\left\{\frac{1}{\lambda_{\mathrm{max}}\left(\bP^{-1}\bQ\right)},\frac{1}{\max\left\{\lambda\left\{\bG\right\}\in\R^{>0}\right\}}\right\}
\end{align}
where $\bP$ and $\bQ$ are defined by \eqref{eq:72} and \eqref{eq:73}, and $\bG\in\C^{M^{2}\times M^{2}}$ is defined by the following block matrix 
\begin{align}
\label{eq:75}
\bG=
\begin{bmatrix}
\bP/2 & -\bQ/2   \\
\bI &  \emph{\textbf{0} }   
\end{bmatrix}
\end{align}
\end{prop}
\begin{proof}
See Appendix \ref{apx:2}.
\end{proof}

\subsection{Transient Analysis}
\label{sub:4C}
\subsubsection{Learning Curves}
Evolution of the EMSE and MSD over time is described by the learning curves which are defined as
\begin{equation} \label{defs}
\zeta(n)= \E{|e_{a}(n)|^{2}},\ \ \eta(n)=\E{\|\bw(n)\|^2}
\end{equation} 
First, the learning curve for EMSE is obtained. Iterating \eqref{eq:23} starting from $n = 0$  we achieve
\begin{align}
\label{eq:76}
\E{\left\|\widetilde{\bw}(n+1)\right\|_{\bsig}^{2}}&=\left\|\bw^{o}\right\|_{{\bmF^{n+1}\bsig}}^{2}+\mu^{2}\s_{\u}^{2}\sum_{j=0}^{n}\bc_{\M}^{\T}\bmF^{k}\bsig\nonumber\\
&\hspace{-2.2cm}=\E{\left\|\widetilde{\bw}(n)\right\|_{\bsig}^{2}}+\left\|\bw^{o}\right\|_{{\bmF^{n+1}\bsig}}^{2}-\left\|\bw^{o}\right\|_{{\bmF^{n}\bsig}}^{2}+\mu^{2}\s_{\u}^{2}\bc_{\M}^{\T}\bmF^{n}\bsig\nonumber\\
&\hspace{-2.2cm}=\E{\left\|\widetilde{\bw}(n)\right\|_{\bsig}^{2}}+\left\|\bw^{o}\right\|_{{\bmF^{n}\left(\bmF-\bI\right)\bsig}}^{2}+\mu^{2}\s_{\u}^{2}\bc_{\M}^{\T}\bmF^{n}\bsig
\end{align}
where $\bw\left(0\right) = 0$. Using $e_a(n)=\bz^{\H}(n)\widetilde{\bw}(n)$ the EMSE at time $n$ can be expressed alternatively as
\begin{align}
\label{eq:77}
\zeta(n) &= \E{\widetilde{\bw}^{\H}(n)\bz(n)\bz^{\H}(n)\widetilde{\bw}(n)}\nonumber\\
&=\E{\E{\widetilde{\bw}^{\H}(n)\bz(n)\bz^{\H}(n)\widetilde{\bw}(n)}|\widetilde{\bw}(n)}\nonumber\\
&=\E{\widetilde{\bw}^{\H}(n)\bC_{\bz}\widetilde{\bw}(n)}\nonumber\\
&=\E{\left\|\widetilde{\bw}(n)\right\|_{\bC_{\bz}}^{2}}
\end{align}
where $\bC_{\bz}=\E{\bz(n) \bz^{\H}(n)}$. Thus, by setting $\bsig=\vect{\bC_{\bz}}$ in \eqref{eq:76}, the time evolution of EMSE can be obtained. 
\begin{prop}
\label{prop:4}
\textbf{\emph{(Learning Curve)}}
Consider the same setting of Proposition \ref{prop:3}. Let $\zeta(n)$ denote the time evolution of EMSE, as defined by \eqref{eq:77}. Then, the EMSE learning curve of PU-ACLMS is given by the following recursion over $n \geq 0$:
\begin{multline}
\label{eq:78}
\zeta(n+1)=\zeta(n) +\left\|\bw^{o}\right\|_{{\bmF^{n}\left(\bmF-\bI\right)\vect{\bC_{\bz}}}}^{2} \\
 +\mu^{2}\s_{\u}^{2}\bc_{\M}^{\T}\bmF^{n}\vect{\bC_{\bz}}
\end{multline}
Similarly, using \eqref{defs} the MSD learning curve can be evaluated by setting $\bsig = \vect{\bI}$ in \eqref{eq:76} which gives
\begin{multline}
\label{eq:78b}
\eta(n+1)=\eta(n) +\left\|\bw^{o}\right\|_{{\bmF^{n}\left(\bmF-\bI\right)\vect{\bI}}}^{2} \\
 +\mu^{2}\s_{\u}^{2}\bc_{\M}^{\T}\bmF^{n}\vect{\bI}
\end{multline}
\end{prop}

\subsubsection{Alternative Steady-state Analysis}
Recursion \eqref{eq:67} can be used  to derive the steady-state values for EMSE and MSD by 
setting $n \to \infty$. In this case recursion \eqref{eq:67} becomes
\begin{align}
\label{eq:79}
\E{\left\|\widetilde{\bw}\left(\infty\right)\right\|_{\bsig}^{2}}=\E{\left\|\widetilde{\bw}\left(\infty\right)\right\|_{\bmF\bsig}^{2}}+\mu^{2}\s_{\u}^{2}\bc_{\M}^{\T}\bsig
\end{align}
which is equivalent to
\begin{align}
\label{eq:80}
\E{\left\|\widetilde{\bw}\left(\infty\right)\right\|_{\left(\bI-\bmF\right)\bsig}^{2}}=\mu^{2}\s_{\u}^{2}\bc_{\M}^{\T}\bsig
\end{align}
To specify the EMSE, it is required to evaluate $\E{\left\|\widetilde{\bw}\left(\infty\right)\right\|_{\bsig}^{2}}$ by setting $\bsig = \left(\bI-\bmF\right)^{-1}\vect{\bC_{\bz}}$. In this way, the left-hand side of \eqref{eq:80} is the steady-state EMSE and we have
\begin{equation}
\label{eq:81}
\zeta(\infty)=\mu^{2}\s_{\u}^{2}\bc_{\M}^{\T}\left(\bI-\bmF\right)^{-1}\vect{\bC_{\bz}}
\end{equation}
For sufficiently small step-size $\mu$, expression \eqref{eq:68} can be approximated by
\begin{equation}
\label{eq:82}
\bmF\approx \bI-2\mu(\bC_{\bz_{\M}}\otimes\bI)
\end{equation}
Substituting \eqref{eq:82} into \eqref{eq:81} leads to
\begin{align} \label{emse_ap}
\zeta(\infty)&=\mu^{2}\s_{\u}^{2}\bc_{\M}^{\T}\left(\bI-\bmF\right)^{-1}\vect{\bC_{\bz}}\nonumber\\
&\approx \frac{\mu\s_{\u}^{2}}{2}\bc_{\M}^{\T}\left(\bC_{\bz_{\M}}\otimes\bI\right)^{-1}\vect{\bC_{\bz}}\nonumber\\
&\approx \frac{\mu\s_{\u}^{2}}{2}\bc_{\M}^{\T}\left(\bC_{\bz_{\M}}^{-1}\otimes\bI\right)\vect{\bC_{\bz}}\nonumber\\
&\approx \frac{\mu\s_{\u}^{2}}{2}\bc_{\M}^{\T}\vect{\bC_{\bz}\bC_{\bz_{\M}}^{-1}}\nonumber\\
&\approx \frac{\mu\s_{\u}^{2}}{2}\tr{\E{\bmJ_{\M}(n)\bz(n)\bz^{\H}(n)\bmJ_{\M}(n)}\bC_{\bz}\bC_{\bz_{\M}}^{-1}} 
\end{align}
Under Assumptions \ref{asp:1} and \ref{asp:2}, we have 
\begin{align}
\label{eq:83}
&\E{\bmJ_{\M}(n)\bz(n)\bz^{\H}(n)\bmJ_{\M}(n)}=\bC_{\bz_{\M}}
\\
\label{eq:84}
&\tr{\E{\bmJ_{\M}(n)\bz(n)\bz^{\H}(n)\bmJ_{\M}(n)}\bC_{\bz}\bC_{\bz_{\M}}^{-1}}=\tr{\bC_{\bz}}
\end{align}
replacing \eqref{eq:83} and \eqref{eq:84} in \eqref{emse_ap} yields
\begin{equation}
\label{eq:85}
\zeta(\infty) \approx \frac{\mu\s_{\u}^{2}}{2}\tr{\bC_{\bz}}
\end{equation}
Since $\tr{\bC_{\bz}}= 2\tr{\bC_{\bu}}$, the steady-state EMSE is given by
\begin{equation}
\label{eq:86}
\zeta(\infty) \approx \mu\s_{\u}^{2}\tr{\bC_{\bu}}
\end{equation}
This is the same result as that obtained for the full-update ACLMS algorithm (see \eqref{eq:41b}), which according to Proposition \ref{prop:1} is also valid for partial update schemes. The steady-state MSD can be calculated in a similar way as
\begin{equation}
\label{eq:86b}
\eta(\infty) \approx \mu \s_{\u}^{2}N
\end{equation}

\subsection{Convergence Rate Analysis}
\label{sub:4E}
It is useful to compare the convergence rate of PU-ACLMS algorithms to the full-update ACLMS algorithm. For this, we examine the decay rates of sequential and stochastic schemes under the following assumptions. 
\begin{assumption}\label{asp:3}\
\begin{enumerate}[(i)]
\item The filter length $N$ is a multiple of $\beta$, i.e. $\frac{N}{M}$ is an integer.
\item The augmented covariance matrix of $\bz(n)$, $\bC_{\bz}$, is block-diagonal such that $\sum_{t=1}^{\beta}\bmT_{t}\bC_{\bz}\bmT_{t}=\bC_{\bz}$. The matrix $\bmT_{t}$ is defined by zeroing out some rows in the identity matrix $\bI$ such that $\sum_{t=1}^{\beta}\bmT_{t}=\bI$.
\end{enumerate}
\end{assumption}
Note that, \emph{Assumption} 3-ii is used to make the analysis tractabe. 

First, it is required to restate the evolution equations of the existing algorithms. Thus, for the regular full-update ACLMS algorithm  the recursion \eqref{eq:54} is rewritten as follows:
\begin{align}
\label{eq:87}
\E{\widetilde{\bw}(n+1)}=\left(\bI-\mu{\bC_{\bz}^{*}}\right)\E{\widetilde{\bw}(n)}
\end{align}
Combining $\beta$-iterations of the recursion \eqref{eq:87}, yields the mean of the coefficient error vector update given by
\begin{align}
\label{eq:88}
\E{\widetilde{\bw}\left(n+\beta\right)}&=\prod_{t=1}^{\beta}\left(\bI-\mu{\bC_{\bz}^{*}}\right)\E{\widetilde{\bw}\left(n+t-1\right)}\nonumber\\
&=\left(\bI-\mu{\bC_{\bz}^{\H}}\right)^{\beta}\E{\widetilde{\bw}(n)}
\end{align}

For the sequential PU-ACLMS, the mean-error-update equation is given by
\begin{align}
\label{eq:89}
\E{\widetilde{\bw}\left(n+1\right)}&=\left(\bI-\mu\bmT_{\text{mod}\left(n,\beta\right)+1}{\bC_{\bz}^{*}}\right)\E{\widetilde{\bw}(n)}
\end{align} 
For any number of optional fractions, the update equation \eqref{eq:89} can be written as 
\begin{align}
\label{eq:90}
\E{\widetilde{\bw}\left(n+\beta\right)}&=\prod_{t=1}^{\beta}\left(\bI-\mu\bmT_{\text{mod}\left(n+t,\beta\right)+1}{\bC_{\bz}^{*}}\right)\E{\widetilde{\bw}(n)}
\end{align} 
Combining $\beta$ updates of \eqref{eq:90}, the following relation for evolution of mean-error-update is obtained:
\begin{align}
\label{eq:91}
\E{\widetilde{\bw}\left(n+\beta\right)}&=\left(\bI-\mu{\bC_{\bz}^{*}}\right)\E{\widetilde{\bw}(n)}
\end{align} 
In the stochastic scheme, the following evolution equation conditioned on a choice $\bmS_{t}$, $t=1,\ldots,\beta$, results in:
\begin{align}
\label{eq:92}
\E{\widetilde{\bw}\left(n+1\right)|\bmS_{t}}&=\left(\bI-\mu\bT_{n}{\bC_{\bz}^{*}}\right)\E{\widetilde{\bw}(n)|\bmS_{t}}\nonumber\\
&=\left(\bI-\mu\bmT_{t}{\bC_{\bz}^{*}}\right)\E{\widetilde{\bw}(n)|\bmS_{t}}
\end{align} 
where the matrix $\bT_{n}$ is chosen randomly from $\bmT_{t}$, $t=1,\ldots,\beta$, with equal probability. Averaging \eqref{eq:92} over all choices of $\bmS_{k}$, we get
\begin{align}
\label{eq:93}
\E{\widetilde{\bw}\left(n+1\right)}&=\left(\bI-\frac{\mu}{\beta}{\bC_{\bz}^{*}}\right)\E{\widetilde{\bw}(n)}
\end{align}
Collecting $\beta$ updates of this equation, the mean of coefficient error vector update is 
\begin{align}
\label{eq:94}
\E{\widetilde{\bw}\left(n+\beta\right)}&=\prod_{t=1}^{\beta}\left(\bI-\frac{\mu}{\beta}{\bC_{\bz}^{*}}\right)\E{\widetilde{\bw}\left(n+t-1\right)}\nonumber\\
&=\left(\bI-\frac{\mu}{\beta}{\bC_{\bz}^{*}}\right)^{\beta}\E{\widetilde{\bw}(n)}
\end{align}
From \eqref{eq:88}, \eqref{eq:91}, and \eqref{eq:94}, the rate of decay $r$ for the existing algorithms are given by
\begin{align}
\label{eq:95}
r_{\mathrm{full}}&= \left(\bI-\mu{\bC_{\bz}^{*}}\right)^{\beta}\\
\label{eq:96}
r_{\mathrm{seq}}&=  \left(\bI-\frac{\mu}{\beta}{\bC_{\bz}^{*}}\right)\\
\label{eq:97}
r_{\mathrm{stoch}}&=  \left(\bI-\frac{\mu}{\beta}{\bC_{\bz}^{*}}\right)^{\beta}
\end{align}

\begin{remark}
\label{rmk:6}
The decay rate of full-update ACLMS algorithm is $\beta$-time faster than those of sequential and stochastic PU-ACLMS algorithms.
\end{remark}

\section{Performance Evaluations}
\label{sec:5}
In the first set of simulations, the performance of proposed PU-ACLMS algorithm for the complex-valued identification problem of  a complex-valued correlated or pseudo-correlated second order non-circular Gaussian signal, is given by
\begin{align}
\label{eq:98}
\bu\left(i+1\right)=0.3\bu(n)+0.1\bu^{*}(n)+q(n)
\end{align} 
where $\bq(n)$ is a zero-mean doubly-white circular Gaussian noise process with unit variance, $\sigma_{q}^{2}=1$, and complementary variance, $\tilde{\sigma}_{q}^{2}=0.9$, giving a high degree of non-circularity. The weights were initialized randomly. All graphs are obtained by  averaging over $100$ independent trials. Fig. \ref{fig:2} show the scatter plot of desired signal $d(n)$ and measurement noise $\u(n)$. 

Figs. \ref{fig:3} and \ref{fig:4} illustrate the resulting theoretical and simulated learning curves, EMSE, and MSD curves of PU-ACLMS algorithm using sequential and stochastic schemes for different choices of coefficients update $M$. From these figures it can be observed that the theoretical analysis precisely specify the EMSE and MSD evolutions, in both transient and steady-state stages.  
Evidently, the full-update cases have the fastest convergence rate in comparison to PU-ACLMS algorithms. Generally speaking, the speed of convergence reduces proportionally for both sequential and stochastic schemes as the number of updated coefficients per iteration is decreased. 
Fig. \ref{fig:5} provides a comparison between MSD curves of sequential and stochastic schemes. As can be seen, the transient and steady-state behavior of of these two schemes are quite similar throughout all stage of adaptation.  

In the next scenario, the performance of PU-ACLMS algorithm for different values of step-sizes, $\mu$ is evaluated. Fig. \ref{fig:6} shows the steady-state EMSE and MSD values of PU-ACLMS algorithm  as a function of step-size using the theoretical expressions and simulated ones. It can be seen from \ref{fig:6} that the theoretical steady-state EMSE and MSD can well follow  the simulated ones. Again in all figures, a good agreement between analytical and empirical results is achieved  with maximum performance difference of only 0.002 dB.

\begin{figure} [t]
\centering \includegraphics [width=8.5cm]{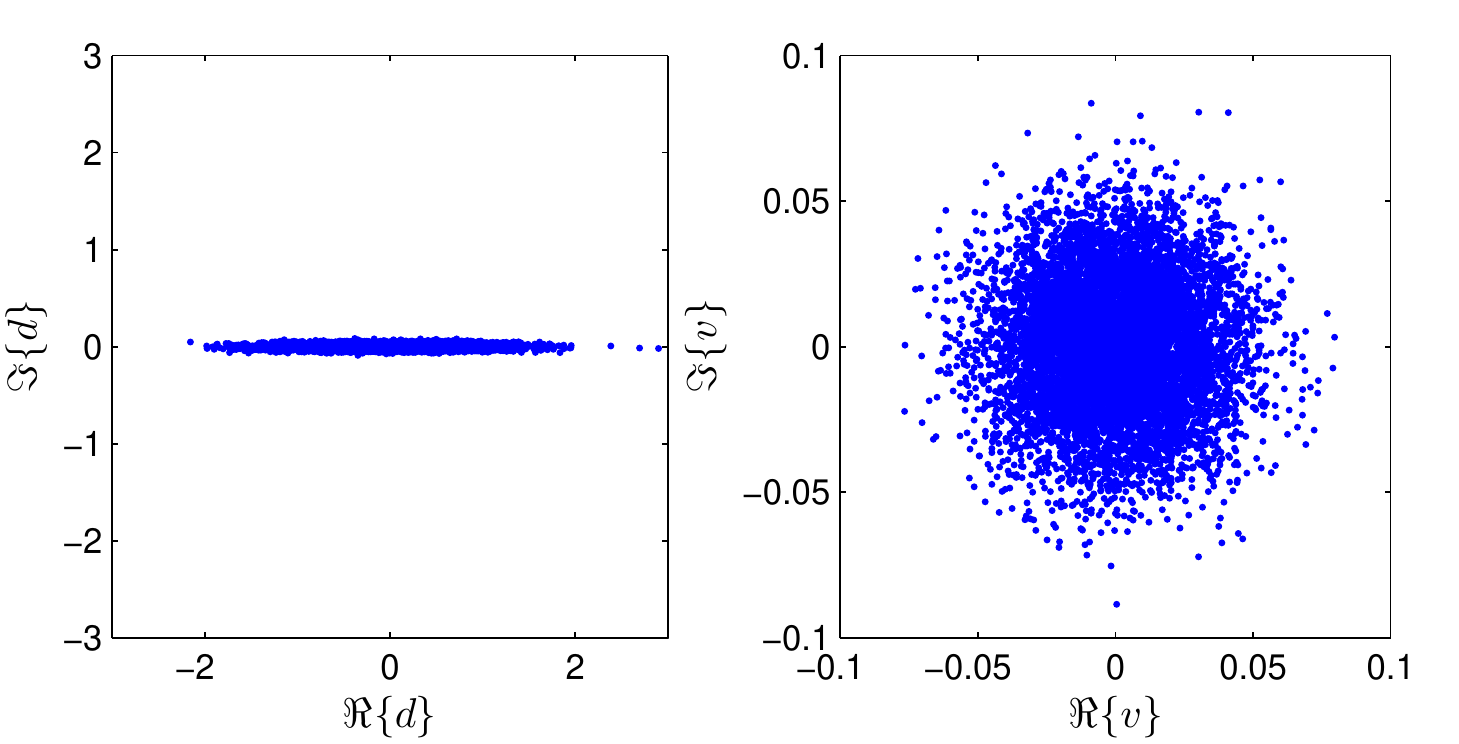} 
\centering \caption{The scatter plot of desired signal $d(n)$ and measurement noise $\u(n)$.}
\label{fig:2}
\end{figure}

\begin{figure} [t]
\centering \includegraphics [width=9cm]{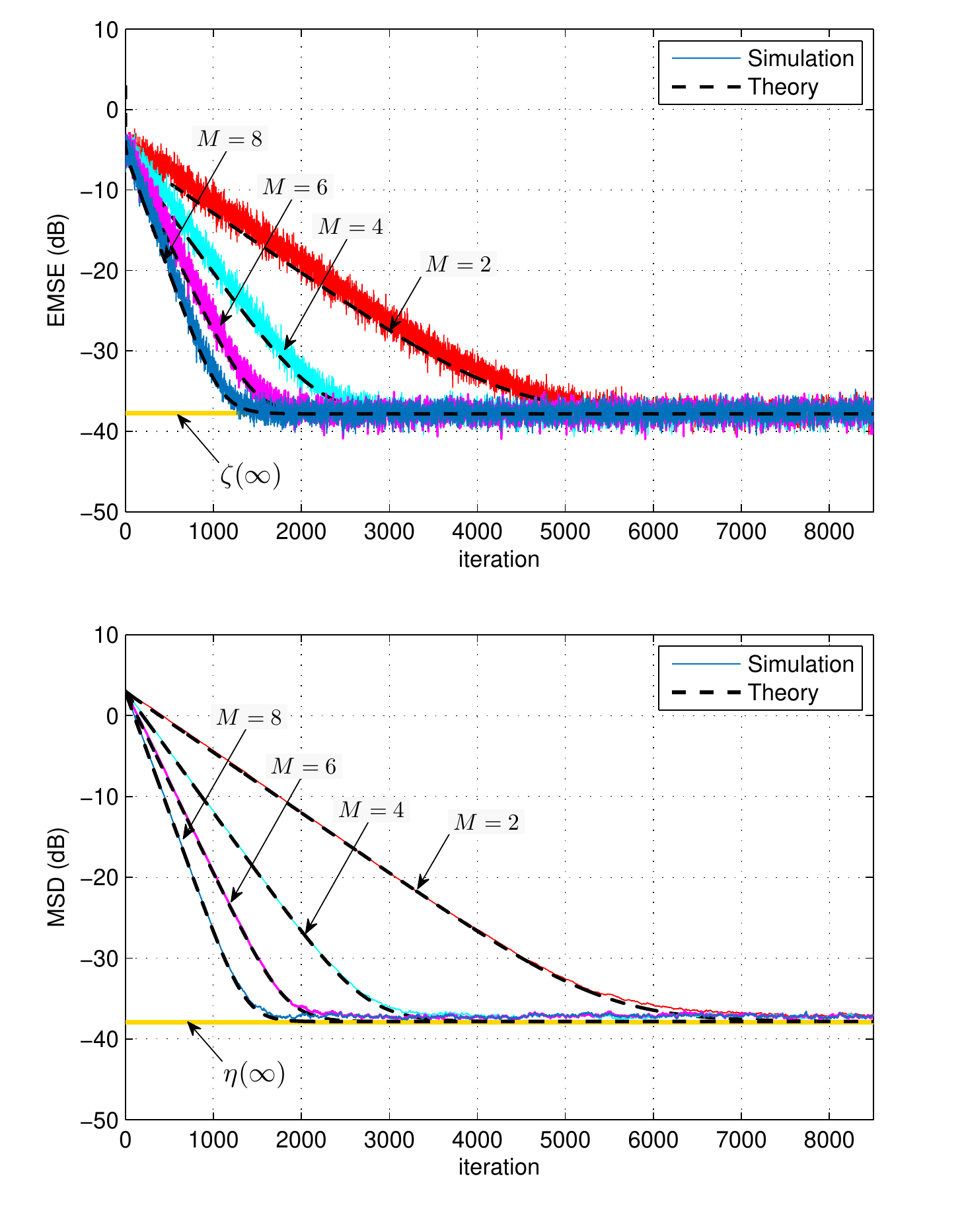} 
\centering \caption{Theoretical and experimental EMSE (top) and MSD (bottom) of PU-ACLMS algorithm using sequential scheme for different values of coefficients update $M$, when the step-size is $\mu = 0.02$.}
\label{fig:3}
\end{figure}

\begin{figure} [t]
\centering 
\includegraphics [width=9cm]{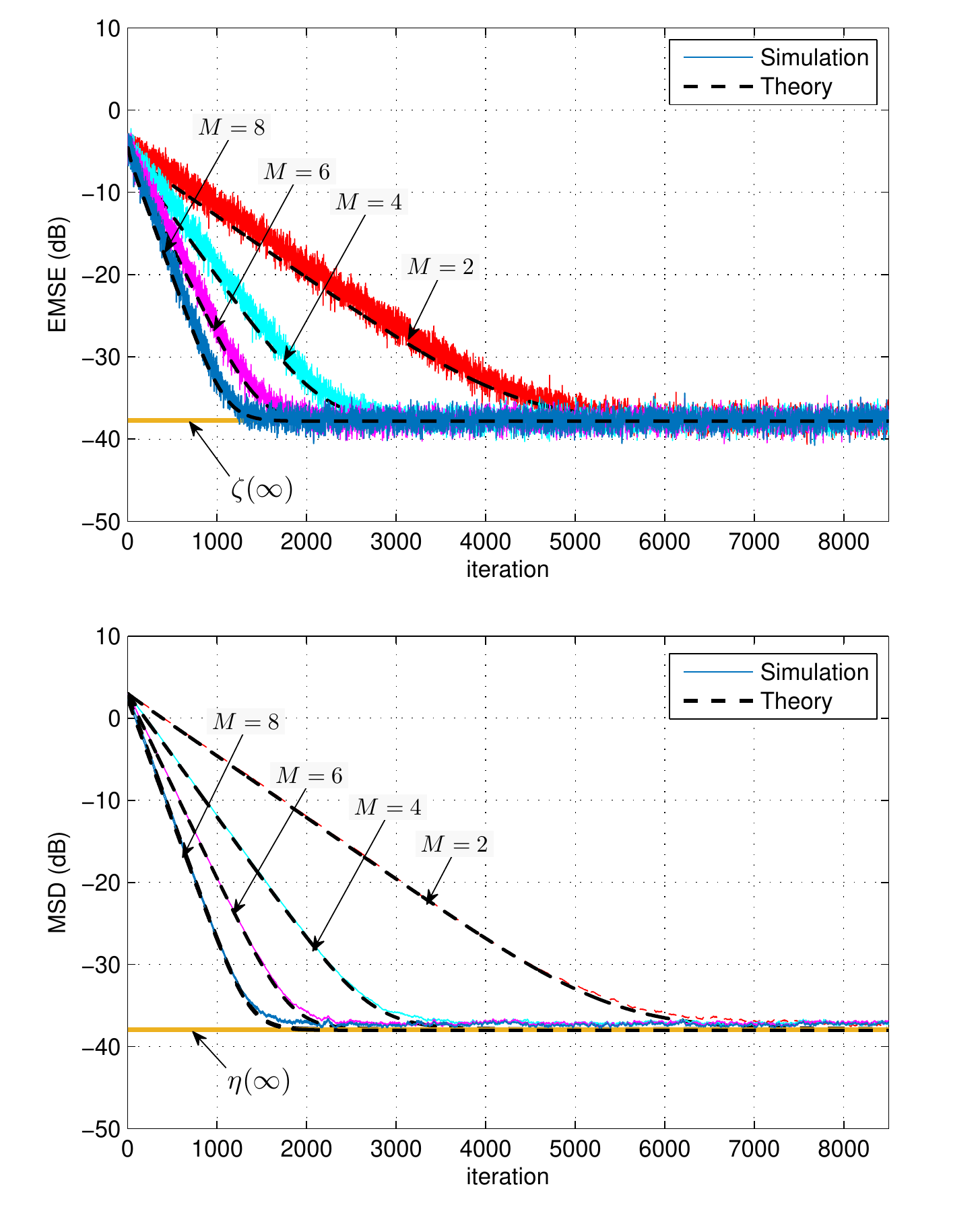} 
\centering \caption{Theoretical and experimental EMSE (top) and MSD (bottom) of PU-ACLMS algorithm using stochastic scheme for different values of coefficients update $M$, when the step-size is $\mu = 0.02$.}
\label{fig:4}
\end{figure}

\begin{figure} [t]
\centering 
\includegraphics [width=9cm]{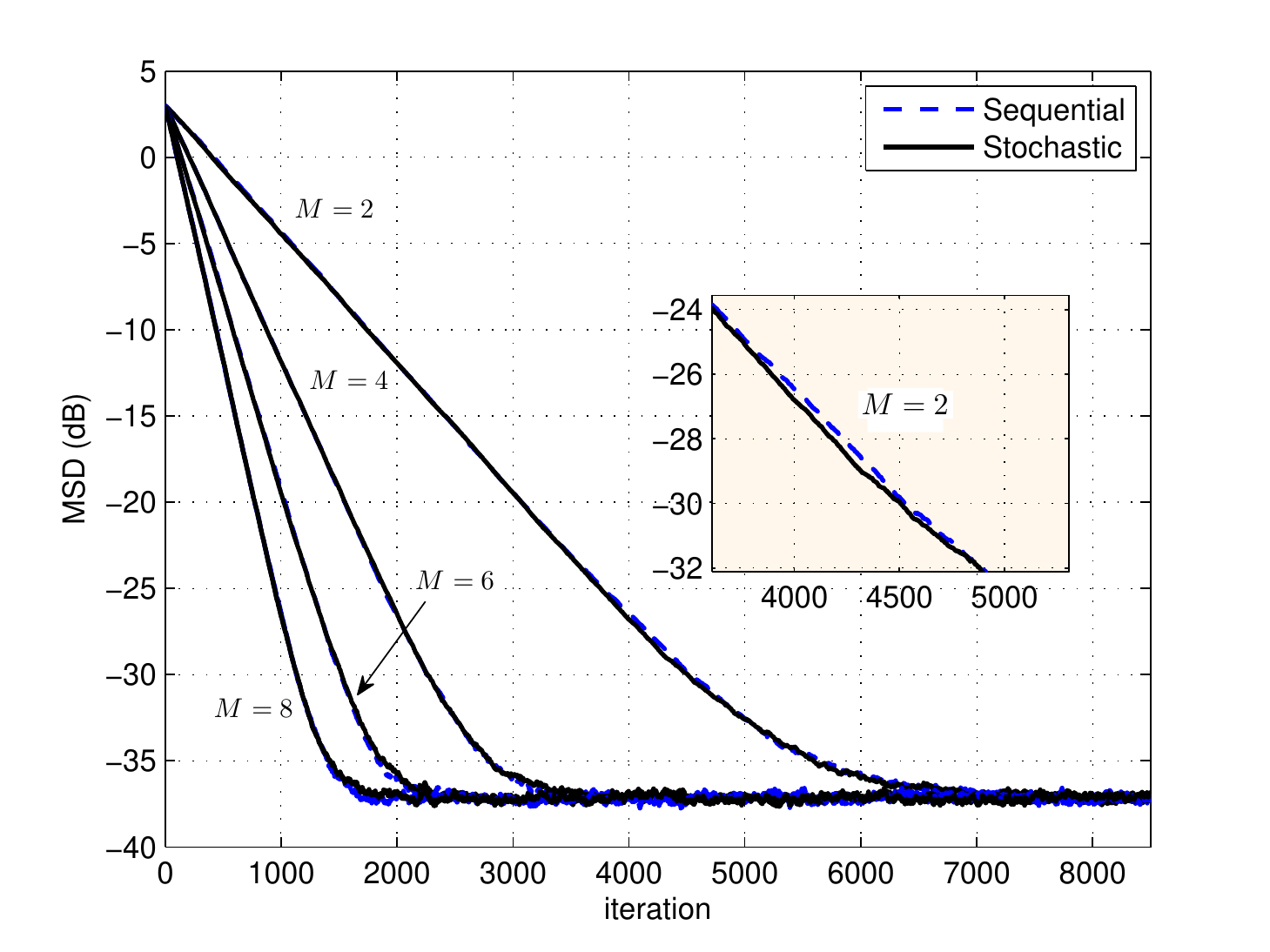} 
\centering \caption{MSD learning curves of PU-ACLMS algorithm for both sequential and stochastic schemes for different values of coefficients update $M$, when the step-size is $\mu = 0.02$.}
\label{fig:5}
\end{figure}
\begin{figure} [t]
\centering \includegraphics [width=9.5cm]{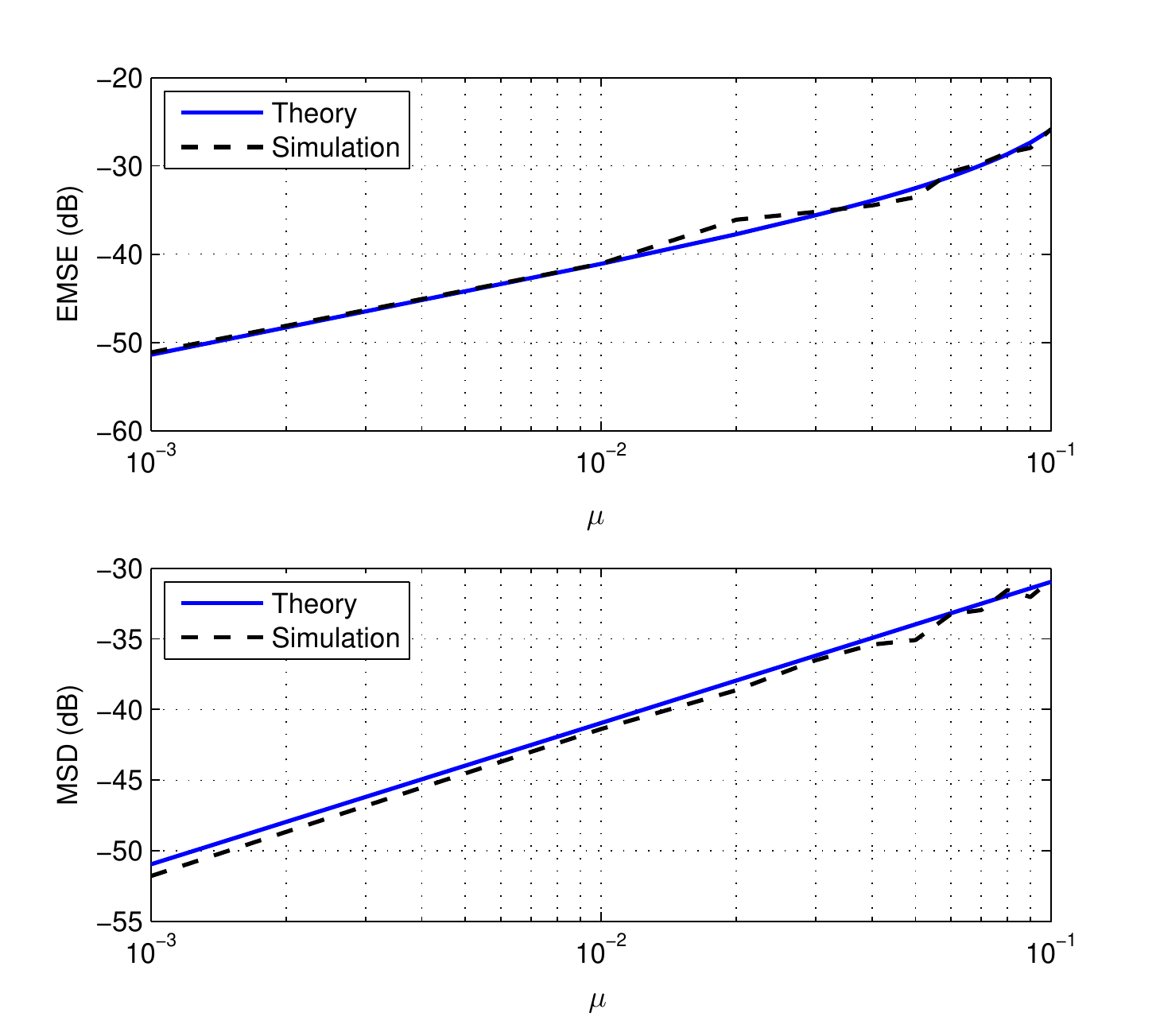} 
\centering \caption{Theoretical and experimental steady-state EMSE (top) and MSD (bottom) of sequential PU-ACLMS algorithm in terms of the step-size parameter $\mu$.}
\label{fig:6}
\end{figure}




\section{Conclusion}
\label{sec:6}
In this paper, a new algorithm, termed briefly PU-ACLMS algorithm has been proposed,  where only a portion of  coefficient weights are selected to update at each iteration. Moreover, the computational complexity for full-update ACLMS and its partial-update implementations have been discussed. It has been concluded that for large filter lengths  
the PU-ACLMS algorithm is able to decrease the  maximum complexity approximately by a factor of two. The steady-state 
performances of the proposed PU-ACLMS algorithm have been analyzed and the mean and mean-square convergence behavior of the PU-ACLMS algorithm   studied for a second-order non-circular Gaussian input regressor. In addition, three cases of 
coefficient weight updates, i.e., full-update, sequential and stochastic partial update have been analyzed.  The EMSE and MSD learning curves have been derived. The provided analysis demonstrated that at a given step-size $\mu$ the level of steady-state EMSE of full-update 
case and PU-ACLMS algorithm are identical. However, their convergence speeds are decreased in  proportion to the number of coefficients weights updated per iteration 
divided by the filter length. The simulation results support the theoretical derivations.

\appendices

\section{Proof of Proposition \ref{prop:2}}
\label{apx:1}
\begin{proof}
Consider the update equation \eqref{eq:54}. Since $\bC_{\bz_{\M}}^{*}$ is Hermitian and positive semi-definite, we can write the eigen-decomposition of it as
\begin{equation}
\bC_{\bz_{\M}}^{*}=\bmU \Lambda\bmU^{\H}\nonumber
\tag{A.1}
\end{equation}
where $\Lambda=\Diag{\lambda_{1},\lambda_{2},\ldots,\lambda_{\M}}$, $\lambda_{k}\geq\lambda_{k+1}$, is the diagonal matrix of the eigenvalues of $\bC_{\bz_{\M}}^{*}$ and $\bmU$ is the matrix of the corresponding eigenvectors \footnote{The eigenvectors, $\bmU$, may be chosen to be orthonormal in which cases $\bmU$ is unitary matrix; $\bmU\bmU^{\H}=\bmU^{\H}\bmU=\bI$}. 

To express the convergence modes \cite{mandic2009complex} solely in terms of the corresponding eigenvalues of the partial correlation matrix, we rotate the partial weight error vector $\widetilde{\bw}(n)$ by the eigen-matrix $\bmU$, that is, $\widetilde{\bw}^{'}(n)=\bmU\widetilde{\bw}(n)$. This linear transformation changes the evolution of the partial weighted error vectors as follows:
\begin{align}
\E{\widetilde{\bw}'\left(n+1\right)}=\left(\bI-\mu\Lambda\right)\E{\widetilde{\bw}'(n)}\nonumber
\tag{A.2}
\end{align}
As the maximum eigenvalue, $\lambda_{\mathrm{max}}$, yields the fastest mode of convergence, the condition for convergence in mean of the PU-ACLMS algorithm becomes
\begin{align}
 \left|\lambda_{\mathrm{max}}\left(\bI-\mu\Lambda\right)\right|\nonumber
\tag{A.3}
\end{align}
Using the fact that $\bC_{\bz_{\M}}^{*}$ and $\bC_{\bz_{\M}}$ have the same eigenvalues, the asymptotic unbiasedness and mean stability of the PU-ACLMS are guaranteed if 
\begin{align}
\left|1-\mu\lambda_{\mathrm{max}}\left(\bC_{\bz_{\M}}\right)\right|<1\nonumber
\tag{A.4}
\end{align}
This inequality determines the stability bounds for step-size $\mu$ as
\begin{align}
0<\mu<\frac{2}{\lambda_{\mathrm{max}}\left(\bC_{\bz_{\M}}\right)}\nonumber
\tag{A.5}
\end{align}
\end{proof}

\section{Proof of Proposition \ref{prop:3}}
\label{apx:2}
\begin{proof}
Consider the matrix $\bmF$ of the form \eqref{eq:71} with positive-definite matrix $\bP>0$, non-negative definite matrix $\bQ \geq 0$, and positive step-size $\mu$. For mean-square stability, we need to determine a bound on $\mu$ such that all eigenvalues of $\bmF$ are strictly inside the unite disc, i.e. $\left|\lambda\left(\bmF\right)\right|<1$. This condition is satisfied if the step-size parameter $\mu$ satisfy
\begin{align}
\max_{\left\|\bx\right\|=1}\left\{\bx\bmF\left(\mu\right)\bx^{H}\right\}<1\nonumber
\tag{B.1}
\end{align}
or equivalently, $\bP-\mu\bQ > 0$, and
\begin{align}
\min\left\{\bx\bH\left(\mu\right)\bx^{H}\right\} > -1\nonumber
\tag{B.2}
\end{align}
or, equivalently, $\bH\left(\mu\right) = 2\bI-\mu\bP+\mu^{2}\bQ > 0$. These inequalities satisfy the stability conditions $\lambda\left(\bmF\right) < 1$ and $\lambda\left(\bmF\right) > -1$, respectively.

In the light of Appendix 25.A of \cite{Sayed2008}, the inequality (B.1) holds if, and only if,
\begin{align}
\mu<\frac{1}{\lambda_{\mathrm{max}}\left(\bP^{-1}\bQ\right)}\nonumber
\tag{B.3}
\end{align}

Since the eigenvalues of $\bH$ vary continuously with step-size parameter $\mu$, an upper-bound on $\mu$, $\mu_{max}$, such that $\bH\left(\mu\right) > 0$ is obtained from the roots of $\text{det}\left[\bH\left(\mu\right)\right] = 0$. Employing the block-matrix determinant principle, the determinant of $\bH\left(\mu\right)$ is equal to the determinant of the block-matrix 
 \begin{align}
\bK\left(\mu\right)
=\begin{bmatrix}
2\bI-\mu\bP &  \mu\bQ \\
-\mu\bI &  \bI \\
\end{bmatrix},\nonumber
\tag{B.4}
\end{align} 
Moreover, since
\begin{align}
\bK\left(\mu\right)
=\begin{bmatrix}
2\bI &  \textbf{0} \\
\textbf{0} &  \bI \\
\end{bmatrix}
\left(\begin{bmatrix}
\bI &  \textbf{0} \\
\textbf{0} &  \bI \\
\end{bmatrix}-\mu
\begin{bmatrix}
\bP/2 &  -\bQ/2 \\
\bI &  \textbf{0} \\
\end{bmatrix}\right)\nonumber
\tag{B.5}
\end{align}
the condition $\text{det}\left[\bK\left(\mu\right)\right]=0$ and $\text{det}\left(\bI-\mu\bG\right)=0$ are identical, where
\begin{align}
\bG\triangleq 
\begin{bmatrix}
\bP/2 &  -\bQ/2 \\
\bI &  \textbf{0} \\
\end{bmatrix}\nonumber
\tag{B.6}
\end{align}

In order to guarantee $\lambda\left(\bmF\right) > -1$, the step-size $\mu$ must satisfy the following inequality
\begin{align}
\mu<\frac{1}{\max\left\{\lambda\left\{\bG\right\}\in\R^{>0}\right\}}\nonumber
\tag{B.7}
\end{align}
It follows that the aforementioned results can be merged together to yield the following condition
\begin{align}
0<\mu<\min\left\{\frac{1}{\lambda_{\mathrm{max}}\left(\bP^{-1}\bQ\right)},\frac{1}{\max\left\{\lambda\left\{\bG\right\}\in\R^{>0}\right\}}\right\}\nonumber
\tag{B.8}
\end{align}
The above range of step-size $\mu$ guarantees the stability of $\bmF$.
\end{proof}



\begin{thebibliography}{10}

\bibitem{Sayed2008}
A.~H. Sayed, ``{Adaptive Filters. Hoboken},'' 2008.

\bibitem{jahanchahi2014complex}
C.~Jahanchahi, S.~Kanna, and D.~Mandic, ``Complex dual channel estimation: Cost
  effective widely linear adaptive filtering,'' \emph{Signal Processing}, vol.
  104, pp. 33--42, 2014.

\bibitem{xia2011widely}
Y.~Xia and D.~P. Mandic, ``Widely linear adaptive frequency estimation of
  unbalanced three-phase power systems,'' \emph{IEEE Transactions on
  Instrumentation and Measurement}, vol.~61, no.~1, pp. 74--83, 2011.

\bibitem{xia2017widely}
Y.~Xia, L.~Qiao, Q.~Yang, W.~Pei, and D.~P. Mandic, ``Widely linear adaptive
  frequency estimation for unbalanced three-phase power systems with multiple
  noisy measurements,'' in \emph{2017 22nd International Conference on Digital
  Signal Processing (DSP)}.\hskip 1em plus 0.5em minus 0.4em\relax IEEE, 2017,
  pp. 1--5.

\bibitem{xia2012adaptive}
Y.~Xia, S.~C. Douglas, and D.~P. Mandic, ``Adaptive frequency estimation in
  smart grid applications: Exploiting noncircularity and widely linear adaptive
  estimators,'' \emph{IEEE Signal Processing Magazine}, vol.~29, no.~5, pp.
  44--54, 2012.

\bibitem{li2018augmented}
Z.~Li, Y.~Xia, W.~Pei, K.~Wang, and D.~P. Mandic, ``An augmented nonlinear lms
  for digital self-interference cancellation in full-duplex direct-conversion
  transceivers,'' \emph{IEEE Transactions on Signal Processing}, vol.~66,
  no.~15, pp. 4065--4078, 2018.

\bibitem{li2019cost}
Z.~Li, Y.~Xia, W.~Pei, and D.~P. Mandic, ``A cost-effective nonlinear
  self-interference canceller in full-duplex direct-conversion transceivers,''
  \emph{Signal Processing}, vol. 158, pp. 4--14, 2019.

\bibitem{xia2018performance}
Y.~Xia, S.~C. Douglas, and D.~P. Mandic, ``Performance analysis of the
  deficient length augmented clms algorithm for second order noncircular
  complex signals,'' \emph{Signal Processing}, vol. 144, pp. 214--225, 2018.

\bibitem{mandic2010steady}
D.~P. Mandic, Y.~Xia, and S.~C. Douglas, ``Steady state analysis of the clms
  and augmented clms algorithms for noncircular complex signals,'' in
  \emph{2010 Conference Record of the Forty Fourth Asilomar Conference on
  Signals, Systems and Computers}.\hskip 1em plus 0.5em minus 0.4em\relax IEEE,
  2010, pp. 1635--1639.

\bibitem{widrow1975complex}
B.~Widrow, J.~McCool, and M.~Ball, ``The complex lms algorithm,''
  \emph{Proceedings of the IEEE}, vol.~63, no.~4, pp. 719--720, 1975.

\bibitem{wu2019steady}
Y.~Wu and J.~Ni, ``Steady-state mean-square deviation analysis of the
  zero-attracting clms algorithm with circular gaussian input,'' \emph{IEEE
  Access}, vol.~7, pp. 52\,331--52\,338, 2019.

\bibitem{mandic2009complex}
D.~P. Mandic and V.~S.~L. Goh, \emph{Complex valued nonlinear adaptive filters:
  noncircularity, widely linear and neural models}.\hskip 1em plus 0.5em minus
  0.4em\relax John Wiley \& Sons, 2009, vol.~59.

\bibitem{adali2011complex}
T.~Adali, P.~J. Schreier, and L.~L. Scharf, ``Complex-valued signal processing:
  The proper way to deal with impropriety,'' \emph{IEEE Transactions on Signal
  Processing}, vol.~59, no.~11, pp. 5101--5125, 2011.

\bibitem{clark2010multiband}
P.~Clark, I.~Kirsteins, and L.~Atlas, ``Multiband analysis for colored
  amplitude-modulated ship noise,'' in \emph{2010 IEEE International Conference
  on Acoustics, Speech and Signal Processing}.\hskip 1em plus 0.5em minus
  0.4em\relax IEEE, 2010, pp. 3970--3973.

\bibitem{zarei2016q}
S.~Zarei, W.~H. Gerstacker, J.~Aulin, and R.~Schober, ``I/q imbalance aware
  widely-linear receiver for uplink multi-cell massive mimo systems: Design and
  sum rate analysis,'' \emph{IEEE Transactions on Wireless Communications},
  vol.~15, no.~5, pp. 3393--3408, 2016.

\bibitem{schreier2003second}
P.~J. Schreier and L.~L. Scharf, ``Second-order analysis of improper complex
  random vectors and processes,'' \emph{IEEE Transactions on Signal
  Processing}, vol.~51, no.~3, pp. 714--725, 2003.

\bibitem{picinbono1995widely}
B.~Picinbono and P.~Chevalier, ``Widely linear estimation with complex data,''
  \emph{IEEE Transactions on Signal Processing}, vol.~43, no.~8, pp.
  2030--2033, 1995.

\bibitem{khalili2016quantized}
A.~Khalili, A.~Rastegarnia, and S.~Sanei, ``Quantized augmented complex
  least-mean square algorithm: Derivation and performance analysis,''
  \emph{Signal Processing}, vol. 121, pp. 54--59, 2016.

\bibitem{schober2004data}
R.~Schober, W.~H. Gerstacker, and L.-J. Lampe, ``Data-aided and blind
  stochastic gradient algorithms for widely linear mmse mai suppression for
  ds-cdma,'' \emph{IEEE Transactions on Signal Processing}, vol.~52, no.~3, pp.
  746--756, 2004.

\bibitem{javidi2008augmented}
S.~Javidi, M.~Pedzisz, S.~L. Goh, and D.~P. Mandic, ``The augmented complex
  least mean square algorithm with application to adaptive prediction problems
  1,'' 2008.

\bibitem{xia2010augmented}
Y.~Xia, C.~C. Took, and D.~P. Mandic, ``An augmented affine projection
  algorithm for the filtering of noncircular complex signals,'' \emph{Signal
  Processing}, vol.~90, no.~6, pp. 1788--1799, 2010.

\bibitem{douglas2009widely}
S.~C. Douglas, ``Widely-linear recursive least-squares algorithm for adaptive
  beamforming,'' in \emph{2009 IEEE International Conference on Acoustics,
  Speech and Signal Processing}.\hskip 1em plus 0.5em minus 0.4em\relax IEEE,
  2009, pp. 2041--2044.

\bibitem{xia2010regularised}
Y.~Xia, S.~Javidi, and D.~P. Mandic, ``A regularised normalised augmented
  complex least mean square algorithm,'' in \emph{2010 7th International
  Symposium on Wireless Communication Systems}.\hskip 1em plus 0.5em minus
  0.4em\relax IEEE, 2010, pp. 355--359.

\bibitem{goh2007augmented}
S.~L. Goh and D.~P. Mandic, ``An augmented extended {K}alman filter algorithm
  for complex-valued recurrent neural networks,'' \emph{Neural Computation},
  vol.~19, no.~4, pp. 1039--1055, 2007.

\bibitem{Dogancay2008}
K.~Dogancay, \emph{{Partial-update adaptive signal processing: Design Analysis
  and Implementation}}.\hskip 1em plus 0.5em minus 0.4em\relax Academic Press,
  2008.

\bibitem{douglas1997adaptive}
S.~C. Douglas, ``Adaptive filters employing partial updates,'' \emph{IEEE
  Transactions on Circuits and Systems II: Analog and Digital Signal
  Processing}, vol.~44, no.~3, pp. 209--216, 1997.

\bibitem{godavarti2005partial}
M.~Godavarti and A.~O. Hero, ``Partial update lms algorithms,'' \emph{IEEE
  Transactions on signal processing}, vol.~53, no.~7, pp. 2382--2399, 2005.

\bibitem{aboulnasr1997selective}
T.~Aboulnasr and K.~Mayyas, ``Selective coefficient update of gradient-based
  adaptive algorithms,'' in \emph{1997 IEEE International Conference on
  Acoustics, Speech, and Signal Processing}, vol.~3.\hskip 1em plus 0.5em minus
  0.4em\relax IEEE, 1997, pp. 1929--1932.

\bibitem{douglas1996analysis}
S.~Douglas, ``Analysis and implementation of the max-{NLMS} adaptive filter,''
  in \emph{Conference Record of The Twenty-Ninth Asilomar Conference on
  Signals, Systems and Computers}, vol.~1.\hskip 1em plus 0.5em minus
  0.4em\relax IEEE, 1996, pp. 659--663.

\bibitem{douglas1994family}
S.~C. Douglas, ``A family of normalized {LMS} algorithms,'' \emph{IEEE signal
  processing letters}, vol.~1, no.~3, pp. 49--51, 1994.

\bibitem{werner2004partial}
S.~Werner, M.~L. De~Campos, and P.~S. Diniz, ``Partial-update {NLMS} algorithms
  with data-selective updating,'' \emph{IEEE Transactions on Signal
  Processing}, vol.~52, no.~4, pp. 938--949, 2004.

\bibitem{Arablouei2014}
R.~Arablouei, K.~Dogancay, S.~Werner, and Y.-F. Huang, ``{Adaptive distributed
  estimation based on recursive least-squares and partial diffusion},''
  \emph{Signal Processing, IEEE Transactions on}, vol.~62, no.~14, pp.
  3510--3522, 2014.

\bibitem{Arablouei2014a}
R.~Arablouei, S.~Werner, Y.-F. Huang, and K.~Dogancay, ``{Distributed least
  mean-square estimation with partial diffusion},'' \emph{Signal Processing,
  IEEE Transactions on}, vol.~62, no.~2, pp. 472--484, 2014.

\bibitem{vahidpour2018analysis}
V.~Vahidpour, A.~Rastegarnia, A.~Khalili, W.~M. Bazzi, and S.~Sanei, ``Analysis
  of partial diffusion {LMS} for adaptive estimation over networks with noisy
  links,'' \emph{IEEE Transactions on Network Science and Engineering}, vol.~5,
  no.~2, pp. 101--112, 2018.

\bibitem{vahidpour2017analysis}
V.~Vahidpour, A.~Rastegarnia, A.~Khalili, and S.~Sanei, ``Analysis of partial
  diffusion recursive least squares adaptation over noisy links,'' \emph{IET
  Signal Processing}, vol.~11, no.~6, pp. 749--757, 2017.

\bibitem{vahidpour2019partial}
------, ``Partial diffusion {K}alman filtering for distributed state estimation
  in multiagent networks,'' \emph{IEEE Transactions on Neural Networks and
  Learning Systems}, 2019.

\bibitem{vahidpour2019performance}
V.~Vahidpour, A.~Rastegarnia, M.~Latifi, A.~Khalili, and S.~Sanei,
  ``Performance analysis of distributed {K}alman filtering with partial
  diffusion over noisy network,'' \emph{IEEE Transactions on Aerospace and
  Electronic Systems}, 2019.

\bibitem{Naylor03}
P.~A. {Naylor} and W.~{Sherliker}, ``A short-sort {M-Max} {NLMS} partial-update
  adaptive filter with applications to echo cancellation,'' in \emph{2003 IEEE
  International Conference on Acoustics, Speech, and Signal Processing, 2003.
  Proceedings. (ICASSP '03).}, vol.~5, April 2003, pp. V--373.

\bibitem{Naylor05}
K.~{Doğançay} and P.~A. {Naylor}, ``Recent advances in partial update and
  sparse adaptive filters,'' in \emph{2005 13th European Signal Processing
  Conference}, Sep. 2005, pp. 1--4.

\bibitem{press2007numerical}
W.~H. Press, S.~A. Teukolsky, W.~T. Vetterling, and B.~P. Flannery,
  \emph{Numerical recipes 3rd edition: The art of scientific computing}.\hskip
  1em plus 0.5em minus 0.4em\relax Cambridge university press, 2007.

\end{thebibliography}
\end{document}